\documentclass[12pt,reqno]{amsart}
\usepackage{appendix}
\usepackage{graphics,graphicx,amsmath,amssymb,epsfig,stmaryrd,color,flafter,xspace,lscape,threeparttable,pdfpages}
\usepackage[round]{natbib}
\usepackage{csquotes,amsfonts,amsthm,geometry,array,booktabs,latexsym,caption,subcaption,multirow,tabularx,rotating}
\usepackage[font=small,labelfont=bf]{caption}
\usepackage{changepage,enumitem,mathrsfs}

\DeclareCaptionFont{tiny}{\tiny}
\everymath{\displaystyle}
\RequirePackage{setspace,indentfirst,epsfig,psfrag,ifthen,ifpdf}

\newtheorem{theorem}{Theorem}
\theoremstyle{plain}

\newtheorem{assumption}{Assumption}
\newtheorem{subassumption}{Assumption\normalfont}[assumption]

\newenvironment{assumption*}
 {\ifnum\value{subassumption}=0 \stepcounter{assumption}\fi\subassumption}
 {\endsubassumption}
\newenvironment{assumption+}[1]
 {\renewcommand{\thesubassumption}{#1}\subassumption}
 {\endsubassumption}
\theoremstyle{definition}
\newtheorem{assump}{Assumption}

\newtheorem{definition}{Definition}
\newtheorem*{definition*}{Definition}
\newtheorem{example}{Example}

\newtheorem{lemma}{Lemma}

\newtheorem{proposition}{Proposition}

\numberwithin{equation}{section}
\begin{document}
\title{Vector Copulas}
\author{Yanqin Fan and Marc Henry}
\thanks{The first version is of July 16, 2020. This version is of \today. It
is a significantly revised version of the manuscript that was previously
circulated under the title \textquotedblleft Vector copulas and vector Sklar
theorem\textquotedblright . We are grateful to Xiaohong Chen and three
anonymous referees for very helpful suggestions and to Guillaume Carlier,
Paul Embrechts, Christian Genest, Haijun Li, L\"{u}dger R\"{u}schendorf,
Brendan Pass, Marco Scarsini, seminar participants at the University Carlos
III de Madrid, and participants of the Pacific Interdisciplinary hub on
Optimal Transport Kickoff Event for helpful comments and discussions. We
also thank Moyu Liao, Hyeonseok Park, and Xuetao Shi for excellent research
assistance. Corresponding author: Yanqin Fan: \texttt{fany88@uw.edu}
Department of Economics, University of Washington, Box 353330, Seattle, WA
98195.}
\address{University of Washington and Penn State}

\begin{abstract}
This paper introduces vector copulas associated with multivariate
distributions with given multivariate marginals, based on the theory of
measure transportation, and establishes a vector version of Sklar's theorem.
The latter provides a theoretical justification for the use of vector
copulas to characterize nonlinear or rank dependence between a finite number
of random vectors (robust to within vector dependence), and to construct
multivariate distributions with any given non-overlapping multivariate
marginals. We construct Elliptical and Kendall families of vector copulas,
derive their densities, and present algorithms to generate data from them.
The use of vector copulas is illustrated with a stylized analysis of
international financial contagion.

\vskip20pt

 \vskip10pt

\end{abstract}

\maketitle



\newpage



\section{Introduction}

The cornerstone of copula theory, known as \emph{Sklar's Theorem}, from %
\citet{Sklar:59}, states that (i) for any multivariate distribution function~%
$F$ on~$\mathbb{R}^{K}$, with univariate marginal distribution functions~$%
F_{1}$, ..., $F_{K}$, there exists a copula function~$C$ such that 
\begin{equation}
F\left( x_{1},\ldots ,x_{K}\right) =C\left( F_{1}\left( x_{1}\right) ,\ldots
,F_{K}\left( x_{K}\right) \right) ,  \label{Sklar59}
\end{equation}%
and (ii) given any copula function~$C$ and any collection of univariate
distribution functions $F_{1},\ldots ,F_{K}$, (\ref{Sklar59}) defines a
multivariate distribution function with copula~$C$ and the marginal
distributions $F_{1},\ldots ,F_{K}$. When the marginal distributions are
continuous,~$C$ in part (i) of Sklar's Theorem is the unique copula
associated with~$F$ and it characterizes the rank dependence structure in $F$%
. Moreover, (ii) provides a general approach to constructing multivariate
distributions from univariate ones.

By virtue of Sklar's Theorem, copulas can be used to characterize nonlinear
or rank dependence between random variables, as distinct from marginal
distributional features, to compute bounds on parameters of multivariate
distributions in problems with fixed marginals, and to construct parametric
and semiparametric families of multivariate distributions from univariate
ones. Applications to quantitative finance, particularly risk management,
portfolio choice, derivative pricing, financial contagion and other areas,
where precise measures of dependence are crucial, are well known and
extensively reviewed, see for instance \citet{Embrechts:2009} and references
therein.

Applications to economics, though fewer, have been no less expansive. First,
the characterization of dependence, as distinct from marginal distributional
features, has been instrumental in modeling the propagation of shocks in
contagion models in \citet{BC:2013}, in holding dependence fixed to measure
partial distributional effects in \citet{Rothe:2012}, in identifying and
estimating auction models in \citet{HLP:2012} and \citet{He:2017}. Second,
the copula approach to problems with fixed marginals has allowed the
computation of sharp bounds on various relevant parameters in treatment
effects models with randomized treatment or treatment on observables, in %
\citet{CL:2019}, \citet{FM:2017} and the many references therein. It has
also been applied to problems of data combination, including ecological
inference, see \citet{RM:2007} and \citet{FSS:2014} for instance. Third, the
copula as a modelling and inference tool has been used to discipline
multiple latent variables and multiple dimensions of unobserved
heterogeneity. As such, the copula approach has been applied to sample
selection models, in \citet{Smith:2003}, to regime switching models in %
\citet{FW:2010} and \citet{CFW:2014}, to simultaneous equations with binary
outcomes in \citet{HV:2017}, as well as the modeling of earnings dynamics in %
\citet{BR:2009} and the measure of intergenerational mobility in %
\citet{CHKS:2014}. Fourth, semiparametric econometrics models including both
time series and cross section models have been constructed using flexible
parametric copulas to model contemporaneous dependence structures in
multivariate models or time dependence in univariate time series models, see %
\citeauthor{CF:2006} (\citeyear{CF:2006,CF:2006b}), \citet{CFT:2006}, %
\citet{Patton:2006}, \citet{Beare:2010}, \citet{CWY:2009}, and %
\citet{CKZ:2009} for properties, estimation, and inference in such models.
The list is surely not exhaustive.

In all these applications, the need for a notion of copula that links
multivariate marginals arises naturally. In propagation models, %
\citet{MP:2017} highlight the need to distinguish within-group and
between-group dependence. Models of treatment effects with multivariate
potential outcomes of interest fall in the class of problems with given
multivariate marginals, in cases where the latter are identified. Censored
and limited dependent variables models with clustered latent variables call
for hierarchical modeling, where a copula operates on vectors of latent
variables, each of which can also be modeled with a traditional copula. In
integrated risk management, modeling and measuring risks of portfolios of
several groups of risks will also benefit from a copula-like tool for
linking multivariate marginals, see \citet{EP:2006}.

However, Sklar's Theorem, as stated above, requires that all the marginals
be univariate. Indeed, \citet{GQR:95} shows that for two random vectors, if
the function~$C:\left[ 0,1\right] ^{2}\rightarrow \left[ 0,1\right] $ is
such that~$F\left( x_{1},x_{2}\right) =C\left( F_{1}\left( x_{1}\right)
,F_{2}\left( x_{2}\right) \right) $ defines a~$\left( d_{1}+d_{2}\right) $%
-dimensional distribution function with marginals~$F_{1}$ with support in~$%
\mathbb{R}^{d_{1}\text{ }}$and $F_{2} $ with support in~$\mathbb{R}^{d_{2}%
\text{ }}$for all~$d_{1}$ and~$d_{2}$ such that~$d_{1}+d_{2}\geq 3$, and for
all distribution functions~$F_{1}$ and~$F_{2}$, then~$C\left(
u_{1},u_{2}\right) =u_{1}u_{2}$. Hence, the only possible copula which works
with non-overlapping multivariate marginals is the independence copula. %
\citet{Ressel:2019} generalizes this impossibility result to more than two
random vectors.

The objective of the present work is to circumvent this impossibility
theorem. The paper develops a vector copula that generalizes the traditional
copula to model and characterize nonlinear or rank dependence between a
finite number of random vectors of any finite dimensions. It relies on the
combination of the theory of probability distribution with given overlapping
marginals, particularly \citet{Vorobev:62} and \citet{Kellerer:64}, with the
theory of transport of probability distributions, particularly \citet{RR:90}%
, \citet{Brenier:91} and \citet{McCann:95}. First, we introduce the concept
of a vector copula and establish a vector version of Sklar's Theorem using
vector ranks proposed in \citet{CGHH:2017} as multivariate probability
transforms to remove marginal distributional features. Vector copulas and
the vector Sklar theorem allow the construction of distributions with 
\textit{any} given non overlapping multivariate marginals, thereby
overcoming the weakness of traditional copulas identified in \citet{GQR:95}.
Second, we show that vector copulas are invariant to comonotonic
transformations, where the multivariate notion of comonotonicity is borrowed
from \citet{GH:2012} and \citet{EGH:2012}, and we define comonotonic and
countermonotonic vector copulas extending Fr\'{e}chet extremal copulas.
Third, we construct flexible parametric families of vector copulas including
Elliptical and Kendall vector copulas and provide algorithms for simulating
from them. They reduce to the corresponding classical copulas when all the
marginals are univariate. Using the Vector Sklar Theorem, we construct new
families of multivariate distributions with any fixed non-overlapping
multivariate marginals and Elliptical or Kendall vector copulas. To
illustrate a possible use of vector copulas, we follow \citet{CF:2006b} and %
\citet{MP:2017} and study financial contagion through the evolution of
between-vector dependence before, during and after the 2008 financial
crisis, for a collection of five aggregate stock indices. 


\subsection*{Related literature}

Separate efforts have been carried out to develop dependence measures for
random vectors robust to within-vector dependence on the one hand, and to
construct specific multivariate distributions with given multivariate
marginal distributions, on the other hand. For the former, \citet{MP:2017}
propose a dependence measure between a finite number of random vectors that
is robust to within-vector dependence and apply it to the study of contagion
in financial markets, inter alia. \citet{GSS:2014} propose extensions of
Spearman's rho and Kendall's tau for two random vectors and show that they
are invariant to increasing transformations of each component of the random
vector. As in the case of random variables, these global measures are
insufficient to characterize the complete nonlinear dependence structure
between random vectors for which analogues of copulas are needed. %
\citet{LSS:96} develop a a device to link several random vectors with given
distributions, based on the Knothe-Rosenblatt transform (%
\citet{Rosenblatt:52}, \citet{Knothe:57}). Related ideas are developed in %
\citet{Rusch:85} and Section~1.6 of \citet{Rusch:2010}.


\subsection*{Notation and conventions}

Let~$(d_{1},\ldots ,d_{K})$ be a finite collection of integers and for each~$%
k\leq K$, let~$\mu _{k}$ be the uniform distribution on~$\mathcal{U}%
_{k}:=[0,1]^{d_{k}}$. Let~$P_{X}$ stand for the distribution of the random
vector~$X$. Let~$P$ denote a given distribution on~$\mathbb{R}^{d_{1}}\times
\ldots \times \mathbb{R}^{d_{K}}$ with marginals~$P_{k}$ on~$\mathbb{R}%
^{d_{k}}$, each~$k\leq K$. When stating generic results applying to all~$%
k\leq K$ such as vector quantiles and ranks, we omit the subscript~$k $ from~%
$d_{k},$ $\mu _{k}$, and~$P_{k}$ unless stated otherwise. Following %
\citet{Villani:2003}, we let~$g\#\nu $ denote the \emph{image measure} (or 
\emph{push-forward}) of a measure~$\nu$ by a measurable map~$g:\mathbb{R}%
^{d}\rightarrow \mathbb{R}^{d}$. Explicitly, for any Borel set~$A$, $g\#\nu
(A):=\nu (g^{-1}(A))$. Specifically,~$(Id,g)\#\nu$ denotes the degenerate
distribution of the vector~$(X,g(X))$, where~$X\sim\nu$, where the latter
means that~$X$ has distribution~$\nu$, whereas~$X\overset{d}{=}Y$ means that~%
$X$ and~$Y$ are identically distributed. The product measure of~$\nu_1$ and~$%
\nu_2$ is denoted~$\nu_1\otimes\nu_2$. The notation~$\|x\|$ refers to the
Euclidean norm. Denote~$\mathbb{S}^{d}:=\{x\in \mathbb{R}^{d}:\;\Vert x\Vert
\leq 1\}$ the unit ball,~$\mathcal{S}^{d-1}:=\{x\in \mathbb{R}^{d}:\;\Vert
x\Vert =1\}$ the unit sphere in~$\mathbb{R}^{d}$. The symbol~$\partial $
denotes the subdifferential,~$\nabla$ the gradient,~$D$ is the Jacobian and~$%
D^2$ the Hessian. The domain of a function~$\psi$ is denoted dom$(\psi)$.
The convex conjugate of a convex lower semicontinuous function~$\psi$ is
denoted~$\psi^\ast$. We use the standard convention of calling weak
monotononicity non-decreasing or non-increasing, as the case may be. If~$%
\Gamma: \mathbb{R}^d\rightrightarrows\mathbb{R}^d$ is a correspondence, then~%
$\Gamma^{-1}(A):=\{x\in\mathbb{R}^d\;:\Gamma(x)=y, \mbox{ some }y\in A\}$.
The transpose of a matrix~$A$ is denoted~$A^{\top } $. The collection of
subsets of a set~$S$ is denoted~$2^S$.


\subsection*{Organization of the paper}

Section~\ref{sec:VC} introduces vector copulas and their properties. Section~%
\ref{sec:param} develops parametric families of vector copulas. 
Appendix~\ref{app:OT} reviews some notions in convex analysis, measure
transport and vector quantiles and ranks. Appendix~\ref{app:proofs} collects
proofs of results in the main text.



\section{Vector Copulas and Vector Sklar Theorem}

\label{sec:VC}

The objective of this work is to provide a tool to analyze patterns of
dependence between random vectors, separately from patterns of dependence
within each vector and marginal information. This will be achieved in two
steps. This section will be concerned with isolating dependence between
random vectors from within dependence and marginal information, and
characterizing the former with a mathematical object called \emph{vector
copula}. Section~\ref{sec:param} will be concerned with the development of
parametric families of vector copulas for the purposes of fitting data,
estimating dependence and constructing new probability distributions with
given multivariate marginals.

The motivations and the details of the construction we propose are best
illustrated with the manipulation of a real data set, namely a subset of the
data analyzed in \citet{MP:2017} to highlight patterns of contagion between
regional financial markets. We consider~$5$ stock market indices, namely the
Hang Seng Index (Hong Kong), the Nikkei 225 Index (Japan), the FTSE 100
Index (United Kingdom), the S\&P 500 Index (United States) and the DAX 100
Index (Germany). We consider EGARCH(1,1) residuals from weekly returns for
the period from May 2009 to January 2014 (post financial crisis) and compare
with the periods January 2004-December 2007 (pre-crisis) and January
2008-April 2009 (crisis) respectively, whenever we find it useful for
illustrative purposes. Unless otherwise mentioned, the names of the stock
indices will refer to the data generating processes for the residuals in the
post-crisis period. We shall index data points with~$i$, and~$n$ will denote
the sample size, i.e.,~$209$,~$65$ and~$566$ for the pre-crisis, crisis and
post-crisis periods respectively.


\subsection{Copulas}

\label{sec:copulas}

Consider first a pair of indices, namely the Hang Seng Index and the FTSE
100 Index. Call~$Y_1$ and~$Y_2$ the random variables with continuous
cumulative distribution functions~$F_1$ and~$F_2$, that generate the Hang
Seng and FTSE residual data, respectively. In the scatterplot of the sample~$%
(Y_{1i},Y_{2i})_{i=1}^{n}$, shown in Figure~\ref{fig:scatter2}, patterns of
dependence between~$Y_1$ and~$Y_2$ are blurred by the marginal features of
each. To isolate patterns of dependence from marginal features, each random
variable in the pair is mapped into a uniform through the probability
integral transform. The resulting random vector is~$(F_1(Y_1),F_2(Y_2))$. It
has uniform marginals, and, since~$F_1$ and~$F_2$ are increasing
transformations, rank orderings, hence rank dependence features, are
undisturbed. The scatterplot of the sample of empirical ranks~$(\hat
F_1(Y_{1i}),\hat F_2(Y_{2i)})_{i=1}^{n}$, where~$\hat F_j$ denotes the
empirical cumulative distribution function of~$Y_j$, $j=1,2$, now
illustrates patterns of dependence, such as left and right tail dependence,
isolated from marginal features in Figure~\ref{fig:rank2}.

\begin{figure}[tbp]
\centering
\begin{subfigure}{.3\linewidth}
   \hspace*{-0.8in}
\includegraphics[width=2.5in]{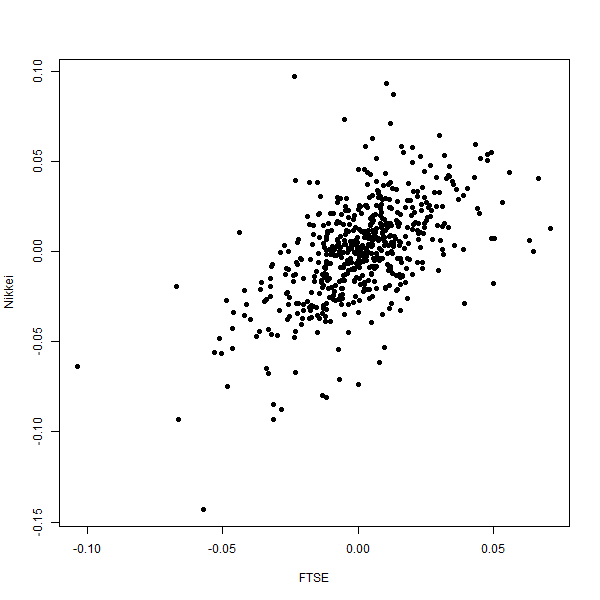} 
\caption{Residuals}
 \label{fig:scatter2}
\end{subfigure}\hspace{0.5cm}%
\begin{subfigure}{.3\linewidth}
   \includegraphics[width=2.5in]{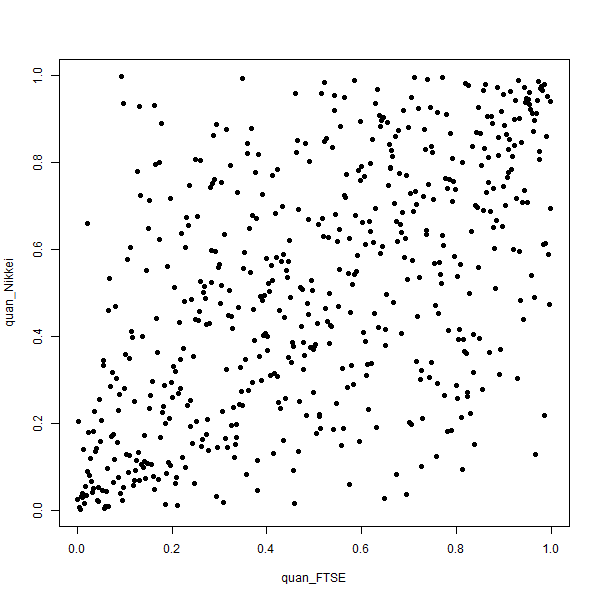} 
\caption{Empirical ranks}
 \label{fig:rank2}
\end{subfigure}
\caption{Scatterplot of residuals (left) and empirical ranks (right).}
\end{figure}

The cumulative distribution function of the random vector~$%
(F_1(Y_1),F_2(Y_2))$ is the \emph{copula} of~$(Y_1,Y_2)$. Conditions for the
claim that the copula characterizes dependence features of~$(Y_1,Y_2)$ are
uniqueness and invariance to transformations of the marginals that leave
ranks unaffected. If~$T_1$ and~$T_2$ are increasing functions, then the
copulas of~$(Y_1,Y_2)$ and~$(T_1(Y_1),T_2(Y_2))$ are identical.

Let $F$ be the cumulative distribution function of~$(Y_1,Y_2)$. Assuming the
marginals are continuous, it follows from Sklar's Theorem that there exists
a unique copula function of~$(Y_1,Y_2)$ or~$F$. It is the distribution
function of the pair of probability integral transforms\footnote{%
For non-continuous marginals, one can use the distributional transform in %
\citet{Rusch:2009} to define a unique copula.} $(F_1(Y_1),F_2(Y_2))$, i.e., 
\begin{eqnarray*}
C\left( u_1,u_2\right) =\Pr \left( F_1(Y_1)\leq u_1,F_2(Y_2)\leq u_2\right)
=F\left( F_1^{-1}\left( u_1\right) ,F_2^{-1}\left( u_2\right) \right) ,
\end{eqnarray*}%
where~$F_{1}^{-1}\left( u_{1}\right)$ and~$F_{2}^{-1}\left( u_2\right)$
denote the quantile functions of~$Y_{1}$ and~$Y_2$. Uniqueness results from
the fact that the probability integral transform is the unique increasing
map pushing forward the marginals $F_{1}$ and~$F_{2}$ to the uniform on $%
\left[ 0,1\right] $. When $F$ is absolutely continuous with pdf $f$, the
copula density function is given by 
\begin{equation}
c\left( u_1,u_2\right) =f\left( F_{1}^{-1}\left( u_{1}\right)
,F_{2}^{-1}\left( u_{2}\right) \right) \prod_{k=1}^{2}\left[ f_{k}\left(
F_{k}^{-1}\left( u_{k}\right) \right) \right]^{-1} .  \label{CopulaD}
\end{equation}

Much of the empirical success of copula theory stems from the ability to
parsimoniously model this dependence separately from marginal features with
parametric families of copulas, and to produce new families of distribution
functions by fitting a parametric copula with any set of marginal
distributions. A common procedure to construct parametric copula families is
best illustrated with the Gaussian copula family. Let~$\Phi_2(x,y;\rho)$ be
the cumulative distribution function of the standard bivariate normal
distribution with correlation coefficient~$\rho$ and let~$\Phi$ be the
cumulative distribution function of the standard univariate normal
distribution. Then, 
\begin{equation*}
C^{\mbox{\scriptsize Ga}}(u_1,u_2;\rho)=\Phi_2(\Phi^{-1}(u_1),%
\Phi^{-1}(u_2);\rho)
\end{equation*}
is the bivariate Gaussian copula with parameter~$\rho$. Popular alternatives
include Student's~$t$ and Archimedean copulas. The densities for the most
common are given in Appendix~\ref{app:cop-list} for convenience.


\subsection{Vector ranks}

\label{sec:vr}

We now turn to the characterization of dependence between two random
vectors. Take~$Y_1$ to now be the bivariate random vector that is assumed to
have generated the residuals from the Hang Seng and the Nikkei indices.
Similarly, take~$Y_2$ to now be the trivariate random vector that is assumed
to have generated the residuals from the FTSE, S\&P and DAX indices. The $%
5\times5$ scatterplot matrix in Figure~\ref{fig:scatter5} shows patterns of
dependence between components of~$Y_1$ and~$Y_2$ that are blurred by
marginal features and within vector dependence, i.e., dependence between the
Hang Seng and the Nikkei, and dependence between the FTSE, the S\&P and the
DAX.

\begin{figure}[tbp]
\centering
\includegraphics[width=4.5in]{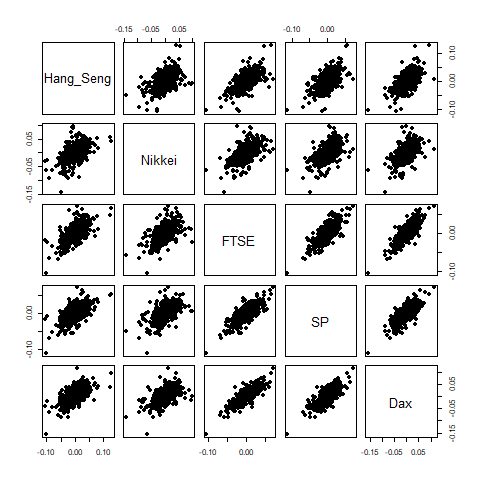}
\caption{Scatterplot matrix of residuals.}
\label{fig:scatter5}
\end{figure}

We pursue the same three objectives as in the previous section: (1) separate
between-vector dependence from within-vector dependence and marginal
features with what we shall call a \emph{vector copula}, (2) show uniqueness
and invariance of the vector copula to transformations of the multivariate
marginals that don't affect between-vector dependence, (3) develop new
families of parametric models to fit multivariate data.

In direct analogy with the bivariate case described in Section~\ref%
{sec:copulas}, we seek to transform~$Y_1$ and~$Y_2$ to uniform random
vectors on~$[0,1]^2$ and~$[0,1]^3$ respectively, in a way that preserves the
structure of dependence between~$Y_1$ and~$Y_2$. The notion of vector rank
proposed in \citet{CGHH:2017}, based on the theory of measure
transportation, has the desired properties, as we now explain. 

The probability integral transform turns a continuous random variable~$W$
into its rank~$F_W(W)$, i.e., a uniform random variable on~$[0,1]$ with the
same rank ordering of outcomes. It is the unique increasing transformation
that turns~$W$ into a uniform random variable on~$[0,1]$. The monotonicity
of the transformation is what preserves rank ordering. Similarly,
Proposition~\ref{prop:polar} below, a seminal result in the theory of
measure transportation in \citet{McCann:95}, states essential uniqueness of
the gradient of a convex function (hence cyclically monotone map)~$R_1$
(resp.~$R_2$) that turns absolutely continuous random vector~$Y_1$ (resp.~$%
Y_2$) to a uniform~$R_1(Y_1)$ on~$[0,1]^2$ (resp.~$R_2(Y_2)$ on~$[0,1]^3$).
See Appendix~\ref{app:OT} for a primer on optimal measure transportation,
cyclical monotonicity and vector ranks.

\begin{proposition}[\citet{McCann:95}]
\label{prop:polar} Let $P$ and $\mu $ be two distributions on $\mathbb{R}%
^{d} $. (1) If $\mu $ is absolutely continuous with respect to the Lebesgue
measure on~$\mathbb{R}^{d}$, with support contained in a convex set $%
\mathcal{U}$, the following holds: there exists a convex function $\psi :%
\mathcal{U}\rightarrow \mathbb{R}\cup \{+\infty \}$ such that $\nabla \psi
\#\mu =P$. The function $\nabla \psi $ exists and is unique, $\mu $-almost
everywhere. (2) If, in addition, $P$ is absolutely continuous on $\mathbb{R}%
^{d}$ with support contained in a convex set $\mathcal{Y}$, the following
holds: there exists a convex function~$\psi ^{\ast }:\mathcal{Y}\rightarrow 
\mathbb{R}\cup \{+\infty \}$ such that $\nabla \psi ^{\ast }\#P=\mu $. The
function $\nabla \psi ^{\ast }$ exists, is unique and equal to $\left(
\nabla \psi \right)^{-1}$, $P$-almost everywhere.
\end{proposition}

Proposition~\ref{prop:polar} is an extension of \citet{Brenier:91} (see also %
\citet{RR:90}). It removes the finite variance requirement, which is
undesirable in our context. Proposition~\ref{prop:polar} is the basis for
the definition of vector quantiles and ranks in \citet{CGHH:2017}. In our
context, it is applied with uniform reference measure.\footnote{%
This vector quantile notion was introduced in \citet{GH:2012} and %
\citet{EGH:2012} and called $\mu $-quantile.} In terms of our empirical
illustration, distribution~$P$ in Definition~\ref{def:MK} below stands for
the distribution of~$Y_1$ (resp.~$Y_2$),~$\mu$ stands for the uniform
distribution on~$[0,1]^2$ (resp.~$[0,1]^3$),~$d=2$ (resp.~$d=3$) and~$%
\nabla\psi^\ast$ is~$R_1$ (resp.~$R_2$).

\begin{definition}[Vector quantiles and ranks]
\label{def:MK} Let $\mu $ be the uniform distribution on~$[0,1]^{d}$, and
let~$P$ be a distribution on $\mathbb{R}^{d}$. The function~$Q:=\nabla\psi$
defined in Proposition~\ref{prop:polar} is called vector quantile associated
with~$P$, or associated with any random vector with distribution~$P$. When~$%
P $ is absolutely continuous, the function~$R:=\nabla\psi^\ast$ defined in
Proposition~\ref{prop:polar} is called vector rank associated with~$P$.
\end{definition}

In case~$d=1$, gradients of convex functions are nondecreasing functions,
hence vector quantiles and ranks of Definition~\ref{def:MK} reduce to
classical quantile and cumulative distribution functions. As the notation
indicates, the function~$\psi^\ast$ of Proposition~\ref{prop:polar} and
Definition~\ref{def:MK} is the convex conjugate of~$\psi$. Here we only
define vector ranks in the absolutely continuous case, and, hence,~$%
\nabla\psi^\ast=(\nabla\psi)^{-1}$, so that ranks and quantiles are inverses
of each other. In case of absolutely continuous distributions~$P$ on~$%
\mathbb{R}^d$ with finite variance, the vector rank function solves a
quadratic optimal transport problem, i.e., vector rank~$R$ minimizes, among
all functions~$T$ such that~$T(Y)$ is uniform on~$[0,1]^d$, the quantity~$%
\mathbb{E }\| Y-T(Y)\|^2$, where~$Y\sim P $. This property, which underlies
cyclical monotonicity of vector ranks and quantiles, as explained in
Appendix~\ref{app:OT}, also has important computational implications, most
notably in allowing the computation of empirical vector ranks using linear
programming.

Return to our illustrative bivariate data generating process~$Y_1$ for the
Hang Seng and Nikkei indices. The notion of vector rank~$R_1(Y_1)$ is best
illustrated with empirical vector ranks~$(\hat R_1(Y_{1i}))_{i=1}^n$. The
latter solve a discrete version of the optimal transport problem of the
previous paragraph. Let~$(u_i)_{i=1}^n$ be a set of regularly spaced points
on~$[0,1]^2$. Define the empirical vector ranks~$(\hat R_1(Y_{1i}))_{i=1}^n$
as a permutation~$(u_{\sigma(i)})_{i=1}^n$ of~$(u_i)_{i=1}^n$ that solves 
\begin{equation*}
\min_\sigma\sum_{i=1}^{n}\|Y_{1i}-u_{\sigma(i)}\|^2
\end{equation*}
among all permutations~$\sigma$ of~$\{1,\ldots,n\}$. The latter is a
discrete version of the optimal transport problem above, hence a linear
programming problem. It is also an \emph{assignment} problem, for which many
efficient algorithms exist in the literature, most notably the \emph{%
Hungarian algorithm} in \citet{Munkres:57} and the \emph{auction algorithm}
in \citet{Bertsekas:88}. Efficient ready-to-use implementations abound.


\subsection{Vector Copulas and Vector Sklar Theorem}

\label{sec:Sklar}

The vector ranks~$R_1$ and~$R_2$ are cyclically monotone transformations
that turn~$Y_1$ and~$Y_2$ into uniform random vectors~$R_1(Y_1)$ and~$%
R_2(Y_2)$ on~$[0,1]^2$ and~$[0,1]^3$ respectively. Hence they play the role
of the probability integral transform for random vectors. The resulting
5-variate random vector~$(R_1(Y_1),R_2(Y_2))$ has uniform bivariate and
trivariate marginals, so that the remaining dependence structure is purely
between-vector dependence. Figure~\ref{fig:rank5} shows the $5\times5$
scatterplot matrix for the sample of empirical ranks~$(\hat R_1(Y_{1i}),\hat
R_2(Y_{2i}))_{i=1}^n$. The off diagonal scatterplots within the diagonal $%
2\times2$ and $3\times3$ blocks look uniform, as expected, given that the
vector ranks~$R_1(Y_1)$ and~$R_2(Y_2)$ are uniform. The other off-diagonal
scatterplots illustrate between-vector dependence, as desired.

\begin{figure}[tbp]
\centering
\includegraphics[width=4.5in]{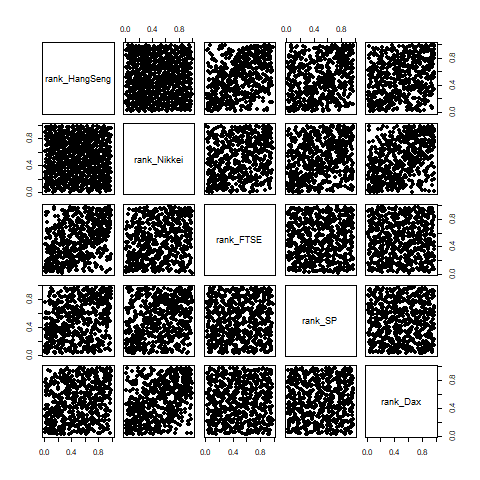}
\caption{Scatterplot matrix of empirical ranks.}
\label{fig:rank5}
\end{figure}

By construction, the distribution of the vector~$(R_1(Y_1),R_2(Y_2))$ is a
probability distribution on~$[0,1]^5$ with uniform multivariate marginals on~%
$[0,1]^2$ and~$[0,1]^3$ respectively. This is all we require to define a
vector copula.

\begin{definition}[Vector Copulas]
\label{def:VC}

A \emph{vector copula}~$C$ is a cumulative distribution function on~$\left[
0,1\right] ^{d}$ with uniform marginals~$\mu _{k}=U[0,1]^{d_k}$, $k\leq K$,
where $d=d_{1}+\ldots +d_{K}$. The associated probability measure~$P_{C} $
will also be referred to as vector copula, when there is no ambiguity.
\end{definition}

Definition~\ref{def:VC} implies that the class of vector copulas is a
subclass of copulas of dimension $d$ having the special feature that the $K$
non-overlapping multivariate marginals are uniform. When $d_{1}=\ldots
=d_{K}=1$, the class of vector copulas is the class of copulas of dimension~$%
K$.

Now that we have settled on a definition for vector copulas, we need to
argue that they characterize between-vector dependence of random vectors
with given multivariate marginals. To this end, we need to show that we can
always associate a vector copula with any random vector, that such a vector
copula is unique under suitable assumptions and that it is invariant to
suitably defined monotonic transformations of the multivariate marginals.
Existence and uniqueness are shown in the following multivariate version of
Sklar's Theorem. Invariance to suitably defined monotonic transformations of
the multivariate marginals will be shown in section~\ref{sec:inv}.

In terms of our empirical illustration, distribution~$P$ in the theorem
below stands for the distribution of vector~$(Y_1,Y_2)$,~$P_j$ stands for
the distribution of~$Y_j$,~$j=1,2$. The vector quantiles are~$%
Q_1=\nabla\psi_1$, and~$Q_2=\nabla\psi_2$, and the vector ranks are~$%
R_1=\nabla\psi_1^\ast$ and~$R_2=\nabla\psi_2^\ast$, ~$K=2$,~$d_1=2 $, and~$%
d_2=3$. Finally,~$C$ is the cumulative distribution function of~$%
(R_1(Y_1),R_2(Y_2))$.

\begin{theorem}[Vector Sklar Theorem]
\label{thm:Sklar} Let~$d_{1},\ldots ,d_{K}$ be a collection of integers and
let~$\mathcal{U}_k=[0,1]^{d_k}$ for each~$k\leq K$. Let~$P$ be any joint
distribution on~$\mathbb{R}^{d_{1}}\times \ldots \times \mathbb{R}^{d_{K}}$
with marginals~$P_{k}$ on~$\mathbb{R}^{d_{k}}$, and~$\psi _{k}$ be the
convex function such that~$\nabla \psi _{k}$ is the vector quantile
associated with~$P_{k}$, each~$k\leq K$. There exists a vector copula~$C$
such that the following properties hold.

\begin{enumerate}
\item There exists a distribution on $(\mathbb{R}^{d_{1}}\times\ldots\times%
\mathbb{R}^{d_{K}})\times(\mathcal{U}_{1}\times\ldots\times\mathcal{U}_{K})$
with margins~$P$ on $\mathbb{R}^{d_{1}}\times\ldots\times\mathbb{R}^{d_{K}}$%
,~$P_C$ on $\mathcal{U}_{1}\times\ldots\times\mathcal{U}_{K}$, and~$(%
\mbox{Id},\nabla\psi_{k})\#\mu_{k}$ on~$\mathcal{U}_{k}\times\mathbb{R}%
^{d_{k}}$, each~$k\leq K$.

\item For any collection~$(A_{1},\ldots ,A_{K})$, where~$A_{k}$ is a Borel
subset of~$\mathbb{R}^{d_{k}}$, $k\leq K$, 
\begin{equation}
P\left( A_{1}\times \ldots \times A_{K}\right) =P_{C}\left(
\partial\psi_{1}^{\ast }\left( A_{1}\right) \times \ldots \times
\partial\psi_{K}^{\ast }\left( A_{K}\right) \right).  \label{Sklar}
\end{equation}

\item If for each~$k\leq K$, $P_{k}$ is absolutely continuous on $\mathbb{R}%
^{d_{k}}$ with support in a convex set, then~$C$ is the unique vector
copula, such that for all Borel sets $B_{1},\ldots ,B_{K}$, in $\mathcal{U}%
_{1},\ldots ,\mathcal{U}_{K}$, 
\begin{equation}
P_{C}\left( B_{1}\times \ldots \times B_{K}\right) =P\left(
\nabla\psi_{1}\left( B_{1}\right) \times \ldots \times \nabla\psi_{K}\left(
B_{K}\right) \right) .  \label{VCopula}
\end{equation}

\item For any vector copula~$C$ defined in Definition 1 and any
distributions $P_{k}$ on~$\mathbb{R}^{d_{k}}$ with vector quantiles~$%
\nabla\psi_{k}$, each~$k\leq K$, (\ref{Sklar}) defines a distribution on $%
\mathbb{R}^{d_{1}}\times \ldots \times \mathbb{R}^{d_{K}}$\ with marginals~$%
P_{k}$, $k\leq K$.
\end{enumerate}
\end{theorem}

When $d_{1}=\ldots =d_{K}=1$, Theorem~\ref{thm:Sklar} reduces to Sklar's
theorem. The vector Sklar theorem plays the same role as Sklar's Theorem for
multivariate marginals.

First, Part~(1) is the corollary of an important result in the theory of
probability measures with fixed overlapping multivariate marginals, based on
combinatorial arguments: Proposition~\ref{prop:Kel} in the appendix, due to %
\citet{Vorobev:62} and \citet{Kellerer:64}. In the illustrative case with
two multivariate marginals only, the result could be proven more directly
with an appeal to the theory of Markov chains. See for instance a result of
Ionescu Tulcea's in Chapter~V, Section~1 of \citet{Neveu:65}. The idea runs
as follows. Take a probability distribution~$P$ on~$\mathbb{R}^{d_1}\times%
\mathbb{R}^{d_2}$ (the distribution of vector~$(Y_1,Y_2)$). Assume its
multivariate marginals are absolutely continuous for simplicity. Given the
distribution~$P$ on~$\mathbb{R}^{d_1}\times\mathbb{R}^{d_2}$ and given the
degenerate probability distributions on~$[0,1]^{d_1}\times\mathbb{R}^{d_1}$
and on~$\mathbb{R}^{d_2}\times[0,1]^{d_2}$ (the distributions of~$%
(\nabla\psi^\ast_1(Y_1),Y_1)$ and~$(Y_2,\nabla\psi^\ast_2(Y_2))$
respectively), there exists a joint distribution on~$[0,1]^{d_1}\times%
\mathbb{R}^{d_1}\times\mathbb{R}^{d_2}\times[0,1]^{d_2}$ with the prescribed
marginals. Then, the latter's marginal on~$[0,1]^{d_1}\times[0,1]^{d_2}$ is
the required distribution~$P_C$ and its cumulative distribution function is
the copula~$C$.

Part~(1) of Theorem~\ref{thm:Sklar} shows existence of a vector copula
associated with any random vector. This motivates the following definition:

\begin{definition}[Vector copula associated with a random vector]
\label{def:VCa} A vector copula formally derived in Theorem~\ref{thm:Sklar}
from a distribution~$P$ with marginals~$P_{k}$ on~$\mathbb{R}^{d_{k}}$, $%
k\leq K$, will be called a \emph{vector copula associated with}~$P$, or 
\emph{vector copula associated with} a random vector with distribution~$P$.
\end{definition}

In our definition of vector copulas associated to a distribution~$P$, we
only rely on the subdifferential~$\partial\psi^\ast$, and the relation~(\ref%
{eq:subdiff}), for each of the multivariate marginals. However, uniqueness
of the vector copula is only guaranteed under absolute continuity of the
multivariate marginals, as is the case for (classical) copulas.

Second, Part~(2) shows that a vector copula associated with~$P$ does indeed
measure the between-dependence structure in~$P$. To see this, let~$%
Y=(Y_{1},\ldots ,Y_{K})$ be a random vector with distribution~$P$ and let
each~$Y_{k}$, $k\leq K$ follow the multivariate marginal distribution~$P_{k}$%
. For each~$k\leq K$, let~$\nabla \psi _{k}$ be the vector quantile
associated with~$P_{k}$. Suppose that for each~$k\leq K$, $P_{k}$ is
absolutely continuous on $\mathbb{R}^{d_{k}}$ with support in a convex set.
Then, from Definition~\ref{def:MK},~$\nabla \psi _{k}^{\ast }\#P_{k}=\mu
_{k} $\ for each~$k\leq K$. Since the reference measure~$\mu _{k}$ is an
independence measure for each~$k\leq K$, the (classical) copula function of~$%
\nabla \psi _{k}^{\ast }\left( Y_{k}\right) $ is the independence copula and
hence the joint distribution of~$\left( \nabla \psi _{1}^{\ast }\left(
Y_{1}\right) ,\ldots ,\nabla \psi _{K}^{\ast }\left( Y_{K}\right) \right) $,
i.e., the vector copula associated with~$P$, measures the between-dependence
structure in~$P$.

Third, Part~(3) of the vector Sklar theorem, or (\ref{VCopula}), provides a
general approach to computing vector copulas of multivariate distributions.
In fact, for absolutely continuous marginals~$P_{k}$ with density~$f_{k}$,
the \emph{Monge Amp\`{e}re Equation} (\ref{eq:MA3}) gives for each~$k\leq K$%
, 
\begin{equation*}
\det \left( D^2\psi_{k}\left( u_{k}\right) \right) =\left[ f_{k}\left(
\nabla\psi_{k}\left( u_{k}\right) \right) \right]^{-1}\text{ for almost
every }u_{k}\in \lbrack 0,1]^{d_{k}}.
\end{equation*}%
We therefore obtain the following expression for the vector copula density~$%
c $, in terms of the density~$f$ of distribution~$P$ (when it exists): 
\begin{eqnarray}
c\left( u_{1},\ldots ,u_{K}\right) &=&f\left( \nabla\psi_{1}\left(
u_{1}\right) ,\ldots ,\nabla\psi_{K}\left( u_{K}\right) \right)
\prod_{k=1}^{K}\det \left( D^2\psi_{k}\left( u_{k}\right) \right)  \notag \\
&=&f\left( \nabla\psi_{1}\left( u_{1}\right) ,\ldots ,\nabla\psi_{K}\left(
u_{K}\right) \right) \prod_{k=1}^{K}\left[ f_{k}\left( \nabla\psi_{k}\left(
u_{k}\right) \right) \right]^{-1}.  \label{VCopulaD}
\end{eqnarray}%
Expression~(\ref{VCopulaD}) extends the copula density in~(\ref{CopulaD}) to
multivariate marginals with the vector quantile~$\nabla\psi_{k}$ replacing~$%
F_{k}^{-1}$ in~(\ref{CopulaD}).

Finally, Part~(4) of the vector Sklar theorem provides a way of constructing
distributions with given non-overlapping marginal distributions of any
finite dimensions. Specifically, it states that for any distributions~$P_{k}$
on~$\mathbb{R}^{d_{k}}$ with vector quantiles~$\nabla\psi_{k}$, each~$k\leq
K $, 
\begin{equation*}
A_{1}\times \ldots \times A_{K}\mapsto P_{C}\left( \partial\psi_{1}^{\ast
}(A_{1})\times \ldots \times \partial\psi_{K}^{\ast }(A_{K})\right)
\end{equation*}%
defines a distribution $P$ on $\mathbb{R}^{d_{1}}\times \ldots \times 
\mathbb{R}^{d_{K}}$\ with marginals~$P_{k}$, $k\leq K$, where $C$ is any
vector copula. When $P_{C}$ and $P_{k}$ for each $k\leq K$ are absolutely
continuous, the density function $f$ associated with $P$ is given by 
\begin{equation}
f\left( y_{1},\ldots ,y_{K}\right) =c\left(
\nabla\psi_{1}^{\ast}(y_{1}),\ldots ,\nabla\psi_{K}^{\ast}(y_{K})\right)
\prod_{k=1}^{K}f_{k}\left( y_{k}\right) ,  \label{DensityDec}
\end{equation}%
which is a direct extension of the density decomposition of copula-based
density functions in the univariate case. The above expresses the
multivariate density function as the product of the copula density function
evaluated at the vector ranks and the density function of~$K$ independent
random vectors with marginals~$P_{1},\ldots ,P_{K}$. This can be used to
construct both MLE and two-step estimators of vector copula-based models in
exactly the same way as copula-based models, see \citet{FP:2014} and
references therein.


\subsection{Comonotonic Invariance}

\label{sec:inv}

Return to our bivariate illustration, where~$Y_1$ and~$Y_2$ are random
variables. The claim that the copula characterizes dependence relies on the
fact that it is not affected by transformations of the marginals that leave
rank ordering intact. In other words,~$(T_1(Y_1),T_2(Y_2))$ has the same
copula as~$(Y_1,Y_2)$ if~$T_1$ and~$T_2$ are increasing functions. Now~$%
\tilde Y_1:=T_1(Y_1) $ and~$Y_1$ are called \emph{comonotonic} when~$T_1$ is
increasing. Equivalently,~$Y_1$ and~$\tilde Y_1$ are comonotonic if they are
both increasing transformations of the same uniform random variable~$U$ on~$%
[0,1]$. In this case, letting~$U:=F_1(Y_1)$, we have indeed~$Y_1=F_1^{-1}(U)$
and~$\tilde Y_1=T_1(F_1^{-1}(U))$. So~$Y_1$ and~$\tilde Y_1$ are comonotonic
if and only if they have the same ranks. We call the copula \emph{%
comonotonic invariant} because vectors~$(Y_1,Y_2)$ and~$(\tilde Y_1,\tilde
Y_2)$ have the same copula when~$Y_1$ and~$\tilde Y_1$ are comonotonic, and~$%
Y_2$ and~$\tilde Y_2$ are comonotonic. This follows directly from the fact
that~$Y_1$ and~$\tilde Y_1$ (resp.~$Y_2$ and~$\tilde Y_2$) have identical
ranks.

Now, in the multivariate case, where~$Y_1$ and~$Y_2$ are random vectors, we
obtain the same invariance property relative to a suitable extension of the
notion of comonotonicity, where quantile functions~$F_{1}^{-1}$ and~$%
F_{2}^{-1}$ are replaced with vector quantiles~$Q_1$ and~$Q_2$. The
following definition is due to \citet{GH:2012} and \citet{EGH:2012}, where
it is called~$\mu$-comonotonicity.

\begin{definition}[Vector comonotonicity]
\label{def:CO} Random vectors~$Y_{1},\ldots ,Y_{J}$ on~$\mathbb{R}^{d}$ are
said to be \emph{comonotonic} if there exists a uniform random vector~$U$ on~%
$[0,1]^{d}$ such that~$Y_{j}=Q_{j}(U)$ almost surely, where~$Q_{j}$ is the
vector quantile of Definition~\ref{def:MK} associated with the distribution
of~$Y_{j}$, for each~$j\leq J$.
\end{definition}

A related notion, namely $c$-comonotonicity, was proposed by \citet{PS:2010}%
. According to their definition, two random vectors~$Y$ and~$\tilde Y$ are $%
c $-comonotonic if~$\tilde Y=\nabla\psi(Y)$ for some convex function~$\psi$. 
The two notions of multivariate comonotonicity both reduce to traditional
comonotonicity in the univariate case, but they differ in the multivariate
case. For instance, two Gaussian random vectors~$Y$ and~$\tilde
Y=\Sigma^{1/2}Y$ may not be comonotonic according to Definition~\ref{def:CO}%
, so that the invariance result does not apply to them.

We now state properties of vector copulas that relate to vector
comonotonicity. Since two vectors are comonotonic if they have identical
vector ranks, comonotonic invariance is indeed an invariance property of
vector copulas to transformations that leave ranks unchanged, as desired.

\begin{theorem}[Comonotonic invariance]
\label{thm:CI} 
Let random vectors~$(Y_{1},\ldots ,Y_{K})$ with distribution~$P$ and~$(%
\tilde{Y}_{1},\ldots ,\tilde{Y}_{K})$ with distribution~$\tilde{P}$ be such
that $Y_{k}$ and~$\tilde{Y}_{k}$ are comonotonic for each~$k$. Then,~$C$ is
a vector copula associated with~$P$ if and only if it is a vector copula
associated with~$\tilde{P}$.
\end{theorem}

When $d_{k}=1$ for all $k\leq K$, comonotonic continuous random variables
are increasing transformations of each other and Theorem \ref{thm:CI}
reduces to the well-known invariance property of copulas. Theorem~\ref%
{thm:CI} shows the importance of deriving vector copulas to flexibly model
between-vector dependence. Indeed, existing parametric multivariate
distribution families cannot model between-vector dependence without
constraining the within vector dependence and marginals. Conversely, a
change in the within-vector dependence, such as replacing~$Y$ by~$\tilde{Y}%
=\Sigma ^{1/2}Y$, also changes the between-vector dependence as
characterized by the vector copula.
If~$\tilde{Y}:=a+bY$, where $a\in
R^{d}$ and~$b>0$ is a scalar, then~$Y$ and~$\tilde{Y}$ are 
comonotonic. Indeed, if~$\psi$ is a convex function, then~$a+b\nabla\psi$ is the gradient of a convex function. Hence, Theorem~\ref{thm:CI} applies to location transformations.


\subsubsection*{Antitone Transformations}

For two vectors, we can also entertain a notion of countermonotonicity as a
multivariate extension of decreasing transformations of two random variables.

\begin{definition}[Vector Countermonotonicity]
Two random vectors~$Y_{1},Y_{2}$ on~$\mathbb{R}^{d}$ are said to be \emph{%
countermonotonic} if there exists a uniform random vector~$U$ on~$[0,1]^d$
such that~$Y_{1}=Q_{1}(U)$ and $Y_{2}=Q_{2}(1_{d}-U)$ almost surely, where~$%
1_{d}$ is the vector of ones and $Q_{j}$ is the vector quantile of
Definition~\ref{def:MK} associated with the distribution of~$Y_{j}$, for
each~$j=1,2$.
\end{definition}

The lemma below shows that vector copulas for random vectors with
comonotonic and countermonotonic subvectors are related in simple,
predictable ways.

\begin{lemma}
\label{lemma:antitone} 
Let random vectors~$(Y_{1},\ldots ,Y_{K})$ with distribution~$P$ and~$(%
\tilde{Y}_{1},\ldots ,\tilde{Y}_{K})$ with distribution~$\tilde{P}$ be such
that $Y_{k}$ and~$\tilde{Y}_{k}$ are comonotonic for each~$k\leq K_{1}$ and $%
Y_{k}$ and~$\tilde{Y}_{k}$ are countermonotonic for each~$K_1<k\leq K$.
Then, the distribution of~$(U_{1},\ldots ,U_{K})$ is a vector copula
associated with~$P$ if and only if the distribution of~$(U_{1},\ldots
,U_{K_{1}},1_{d_{K_{1}+1}}-U_{K_{1}+1},\ldots ,1_{d_{K}}-U_{K})$ is a vector
copula associated with~$\tilde{P}$.
\end{lemma}


\subsection{Extremal Vector Copulas}

A first step towards modeling dependence with copulas is to model extremes.
We present vector copulas that characterize independence on the one end, and
maximal dependence on the other end.

\begin{enumerate}
\item The independence vector copula has distribution~$\mu_{1}\otimes\cdots%
\otimes\mu_{K}$, i.e., the uniform distribution on~$[0,1]^{d}$, with $%
d:=d_{1}+\ldots +d_{K}$.

\item When~$d_{1}=\ldots =d_{K}$, the comonotonic vector copula is the
vector copula with comonotonic multivariate marginals.

\item When~$K=2$ and~$d_1=d_2$, the countermonotonic vector copula is the
vector copula with countermonotonic multivariate marginals.
\end{enumerate}

We will denote~$C^{I}$, $\overline{C}$ and~$\underline{C}$ the independence,
comonotonic and countermonotonic vector copulas respectively.

\begin{lemma}[Comonotonic Vector Copula]
\label{lemma:CC} Let~$U$ be any uniform random vector on~$[0,1]^d$. Then,
the probability distribution associated with the comonotonic vector copula
is the distribution of~$(U,\ldots ,U)$. In addition, the probability
distribution associated with the countermonotonic vector copula is the
distribution of~$(U,1_{d}-U)$.
\end{lemma}

Hence, for any collection of Borel measurable subsets~$A_{k}\subset \lbrack
0,1]^{d}$, $k\leq K,$ the probability distribution associated with the
comonotonic vector copula takes values~$P_{\overline{C}}(A_{1}\times\ldots%
\times A_{K})=\mu \left( \cap _{k\leq K}A_{k}\right) $ and the
countermonotonic vector copula takes values~$P_{\underline{C}} (A_{1}\times
A_{2})=\mu \left( A_{1}\cap (1_{d}-A_{2})\right) .$ Letting $A_{k}=(0,u_{k}]$
with $u_{k}\in \lbrack 0,1]^{d}$ for each $k\leq K$, we obtain 
\begin{equation*}
~\overline{C}\left( u_{1},\ldots ,u_{K}\right) =\mu \left( (0,\min_{k\leq
K}u_{k}]\right) =\prod _{j=1}^{d}\left( \min_{k\leq K}u_{kj}\right) \text{.}
\label{UpperVC}
\end{equation*}%
When $d=1$, $\overline{C}\left( u_{1},\ldots ,u_{K}\right)
=\min\{u_1,\ldots,u_K\}$, the Fr\'{e}chet upper bound copula. However $%
\overline{C}$ differs from the Fr\'{e}chet upper bound when $d>1$. Since it
is a copula, $\overline{C}\left( u_{1},\ldots ,u_{K}\right) \leq \min_{j\leq
d,k\leq K}u_{kj}$ and strict inequality holds for some $\left( u_{1},\ldots
,u_{K}\right) \in \lbrack 0,1]^{dK}$. Similarly, when $d=1$, 
\begin{eqnarray*}
\underline{C}(u_{1},u_{2}) = \Pr \left( U\leq u_{1},1-U\leq u_{2}\right) =
\max \left( u_{1}+u_{2}-1,0\right)
\end{eqnarray*}%
which is the Fr\'{e}chet lower bound copula. When~$d>1$, $\underline{C}%
(u_{1},u_{2})$ is still a copula although the Fr\'{e}chet lower bound is not.

It follows from Theorem~\ref{thm:Sklar} that for any collection of
absolutely continuous distributions~$P_{k}$ on~$\mathbb{R}^{d}$ with vector
quantiles~$\nabla\psi_{k} $, each~$k\leq K$, the expression below defines a
distribution~$\overline{P}$ on~$\mathbb{R}^{d}\times \ldots \times \mathbb{R}%
^{d}$ with marginals~$P_{k}$, $k\leq K $: 
\begin{equation*}
\overline{P}\left( A_{1}\times \ldots \times A_{K}\right) =P_{\overline{C}%
}\left( \nabla\psi_{1}^{\ast}(A_{1})\times \ldots \times
\nabla\psi_{K}^{\ast}(A_{K})\right) .
\end{equation*}%
The distribution~$\overline{P}$ characterizes comonotonic random vectors
with absolutely continuous marginals~$P_{k}$, $k\leq K$. To show this, first
let $\left( Y_{1},\ldots ,Y_{K}\right) \sim \overline{P}$. Then $\left(
\nabla\psi_{1}^{\ast}(Y_{1}),\ldots ,\nabla\psi_{K}^{\ast}(Y_{K})\right)
\sim P_{\overline{C}}$ and hence there exists a uniform random vector~$U$ on~%
$[0,1]^d$ such that~$\nabla\psi_{k}^{\ast}(Y_{k})=U$ for all~$k\leq K$.
Conversely, let~$\left( Y_{1},\ldots ,Y_{K}\right) $ be comonotonic random
vectors with absolutely continuous marginals~$P_{k}$, $k\leq K$. By
definition, there exists a uniform random vector~$U$ on~$[0,1]^d$ such that $%
\nabla\psi_{k}^{\ast}(Y_{k})=U$ for all $k\leq K$. Hence 
\begin{eqnarray}
\Pr \left( Y_{1}\in A_{1},\ldots ,Y_{K}\in A_{k}\right) &=&\Pr \left(
\nabla\psi_{1}(U)\in A_{1},\ldots ,\nabla\psi_{K}(U)\in A_{K}\right)  \notag
\\
&=&\Pr \left( U\in \nabla\psi_{1}^{\ast}\left( A_{1}\right) ,\ldots ,U\in
\nabla\psi_{K}^{\ast}\left( A_{K}\right) \right)  \notag \\
&=&P_{\overline{C}}\left( \nabla\psi_{1}^{\ast}(A_{1})\times \ldots \times
\nabla\psi_{K}^{\ast}(A_{K})\right)  \notag \\
&=&\mu \left( \cap _{k\leq K}\nabla\psi_{k}^{\ast}(A_{k})\right) .
\label{UpperDis}
\end{eqnarray}%
Specifically for $d=1,A_{k}=(-\infty ,x_{k}]$ with $x_{k}\in \mathbb{R}$ for
each $k\leq K$, we have 
\begin{equation*}
F\left( x_{1},\ldots ,x_{K}\right) =\mu \left( \min_{1\leq k\leq
K}F_{k}\left( x_{k}\right) \right) =\min_{1\leq k\leq K}F_{k}\left(
x_{k}\right) ,
\end{equation*}%
which is the well-known Fr\'{e}chet upper bound distribution. However, it is
well known that for $d>1$, the Fr\'{e}chet upper bound $\min_{1\leq k\leq
K}F_{k}\left( x_{k}\right) $ is not a distribution function except for very
specific marginals, see Proposition 5.3 in \citet{Rusch:2010}. In sharp
contrast, the comonotonic vector copula $\overline{C}$ always defines a
distribution function for any marginals through (\ref{UpperDis}).


\subsection{Dependence modeling}

\label{sec:dm}

As noted above, the success of copula theory is partly due to the ability to
parsimoniously model dependence in random vectors with parametric copula
families, and to produce flexible new families of multivariate distributions
by fitting parametric copulas with any given marginals. We show that the
same objective can be achieved with vector copulas, thanks to the vector
Sklar Theorem.

Theorem~\ref{thm:Sklar} implies that given a vector copula~$C$, and given a
set of absolutely continuous multivariate marginal distributions~$P_{k}$ on~$%
\mathbb{R}^{d_{k}}$ with associated vector quantile~$\nabla\psi_{k}$, the
distribution~$P$ defined for Borel sets~$A_{1},\ldots ,A_{K}$, by 
\begin{equation*}
P\left( A_{1}\times \ldots \times A_{K}\right) =P_{C}\left(
\nabla\psi_{1}^{\ast}\left( A_{1}\right) \times \ldots \times
\nabla\psi_{K}^{\ast}\left( A_{K}\right) \right)
\end{equation*}%
is a multivariate distribution with vector copula~$C$ and non-overlapping
marginals~$P_{k}$. Furthermore, if vector copulas~$C$ admits density~$c$,~(%
\ref{DensityDec}) implies that 
\begin{equation*}
f\left( y_{1},\ldots ,y_{K}\right) =c\left(
\nabla\psi_{1}^{\ast}(y_{1}),\ldots ,\nabla\psi_{K}^{\ast}(y_{K}) \right)
\prod_{k=1}^{K} f_{k}\left( y_{k}\right)  \label{Meta-V}
\end{equation*}%
is the density of a multivariate distribution with vector copula~$C$ and
marginal distributions $P_{k}$ for $k\leq K$. Given a parametric vector
copula family~$\{C(.;\theta): \theta\in\Theta\}$, but without parameterizing
the marginal distributions, the above expression results in a semiparametric
multivariate distribution, with finite dimensional vector copula parameter~$%
\theta$ (which measures the between vector dependence) and infinite
dimensional marginal parameters~$f_{k}$, all~$k\leq K $. Combined with
traditional copula modeling of the multivariate marginals to further reduce
dimensionality, vector copulas developed here also allow a flexible
hierarchical approach to multivariate modeling.

Our first family of vector copulas is obtained using Theorem~\ref{thm:Sklar}
from a Gaussian vector, whose multivariate marginals are standard Gaussian
vectors.

\begin{example}[Gaussian Vector Copulas]
\label{ex:Gaussian}

Let~$d_{1},\ldots ,d_{K}$ be a collection of integers, and let 
\begin{equation*}
\Omega =\left( 
\begin{tabular}{llll}
$I_{d_{1}}$ & $\Omega _{12}$ & $\cdots $ & $\Omega _{1K}$ \\ 
$\Omega _{21}$ & $I_{d_{2}}$ & $\cdots $ & $\Omega _{2K}$ \\ 
$\vdots $ & $\vdots $ & $\ddots $ & $\vdots $ \\ 
$\Omega _{K1}$ & $\Omega _{K2}$ & $\cdots $ & $I_{d_{K}}$%
\end{tabular}%
\right) ,
\end{equation*}%
where $\Omega _{ij}$ is a non-degenerate correlation matrix of dimension $%
d_{i}\times d_{j}$ for $i,j=1,..,K$ and $i\neq j$. For $u_{k}\in \left[ 0,1%
\right] ^{d_{k}}$, $k\leq K$, let 
\begin{equation*}
C^{Ga}\left( u_{1},\ldots ,u_{K};\Omega \right) =\Phi _{d}\left( \nabla
\psi_{1}\left( u_{1}\right) ,\ldots ,\nabla \psi_{K}\left( u_{K}\right)
;\Omega \right) ,  \label{GaussianN}
\end{equation*}%
where $d=d_{1}+...+d_{K}$, $\Phi _{d}\left( \cdot ;\Omega \right) $ is the
distribution function of the multivariate normal with zero mean and variance
covariance matrix $\Omega $, and for each~$k\leq K$, 
\begin{eqnarray}  \label{eq:Phi-1}
\nabla \psi _{k}(u_{k})=\Phi^{-1}(u_k):=\left( \Phi ^{-1}(u_{k1}),\ldots
,\Phi ^{-1}(u_{kd_{k}})\right),
\end{eqnarray}
where~$u_k=(u_{k1},\ldots,u_{kd_k})$ and~$\Phi $ is the distribution
function of the standard normal distribution. For each~$k\leq K$, the map~$%
\nabla\psi_k$ is indeed the vector quantile of the multivariate standard
normal by Lemma~\ref{lemma:GS} in Appendix~\ref{app:OT}. The map~$C^{Ga}$ is
a vector copula by Definition \ref{def:VC}. Moreover when $d_{k}=1$ for all $%
k\leq K$,~$C^{Ga}$ reduces to the traditional Gaussian copula.

The following algorithm to simulate from a Gaussian vector copula
generalizes Algorithm~5.9 in \citet{NFE:2005} for the simulation of Gaussian
copulas.

Step 1. Perform a Cholesky decomposition of $\Omega $ to obtain the Cholesky
factor $\Omega ^{1/2}$;

Step 2. Generate a $d$-dimensional standard normal vector $Z$ and set $%
Y=\Omega ^{1/2}Z;$

Step 3. Vector~$U=\left( \Phi \left( Y_{1}\right) ,\ldots ,\Phi \left(
Y_{K}\right) \right)$, where~$Y=(Y_1,\ldots,Y_K),$ is distributed according
to the copula~$C^{Ga}\left( \cdot ;\Omega \right) $.
\end{example}

To illustrate the vector copula approach to multivariate modeling, we fit a
Gaussian vector copula to our 5-dimensional random vector~$(Y_1,Y_2)$ of
residuals. Assume~$(Y_1,Y_2)$ have Gaussian vector copula~$C^{\mbox{\scriptsize Ga}%
}(u_1,u_2;\Omega)$ as in Example~\ref{ex:Gaussian}. The covariance matrix~$%
\Omega$ is estimated from the empirical ranks~$(\hat R_1(Y_{1i}),\hat
R_2(Y_{2i}))_{i=1}^n$ using the method of moments and Figure~\ref%
{fig:simul2x3} shows a $2\times3$ off diagonal scatterplot matrix from a
sample of~$n$ data points simulated independently from the estimated
Gaussian vector copula. Figure~\ref{fig:scatter2x3} replicates the top right
off-diagonal block from the vector rank scatterplot matrix in Figure~\ref%
{fig:rank5}.

\begin{figure}[tbp]
\centering
\begin{subfigure}{.3\linewidth}
  \centering
  \hspace*{-.3in}
  \includegraphics[width=4in]{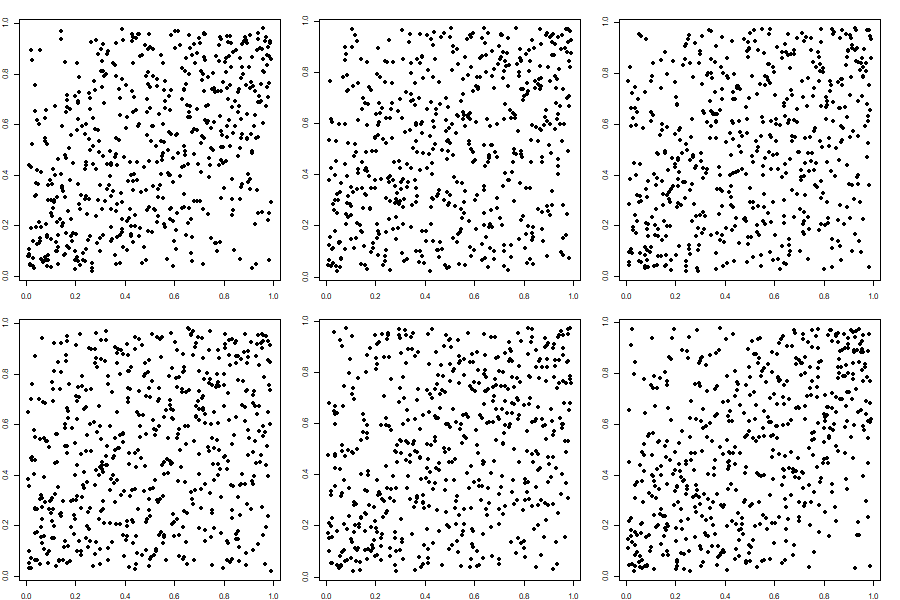} 
\caption{Empirical vector ranks}
 \label{fig:scatter2x3}
\end{subfigure} \newline
\vspace{0.4cm}%
\begin{subfigure}{.3\linewidth}
  \centering
  \hspace*{-1in}
 \includegraphics[width=4in]{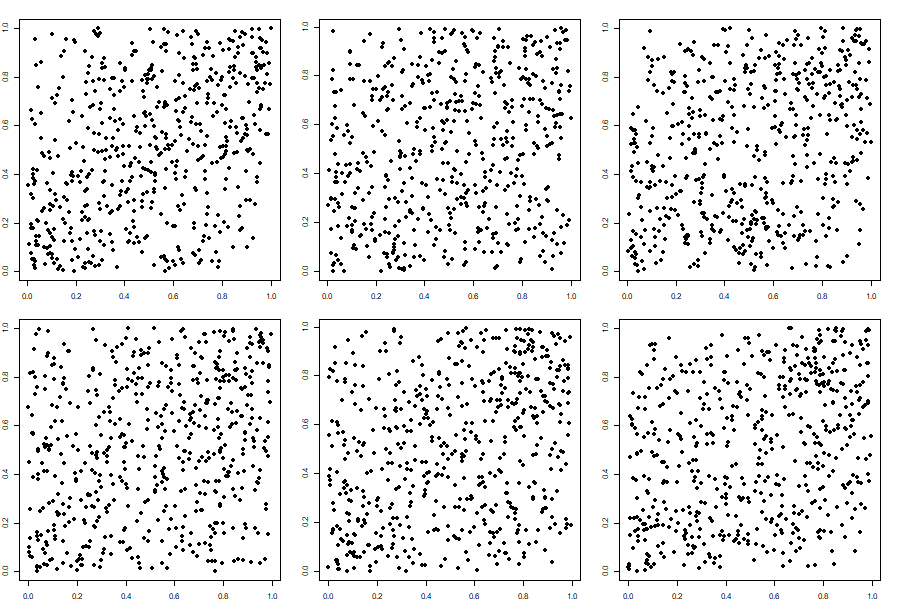} 
 \caption{Gaussian vector copula}
 \label{fig:simul2x3}
\end{subfigure}
\caption{Scatterplot of empirical vector ranks (rows are Nikkei and FTSE,
columns are Hang Seng, S\&P and DAX) and data simulated from a Gaussian
vector copula.}
\end{figure}



\section{Parametric Vector Copula Families}

\label{sec:param}

This section introduces two general classes of parametric vector copulas.
Similarly to Gaussian vector copulas constructed in the previous section,
the first class of vector copulas is constructed via the multivariate
analogue of the inversion method from elliptical distributions. However,
since vector ranks of general elliptical distributions do not have closed
form expressions, we introduce compositions of a finite number of McCann's
measure transport maps and call them \emph{composition measure transports}.
These maps will be used to construct parametric vector copulas from
elliptical distributions. The same principle can be extended to other
parametric families of distributions. The second class of vector copulas we
develop here, called \emph{Kendall vector copulas} is derived from a
stochastic representation of vector copulas in Proposition~\ref{SD-VC}.


\subsection{Composition Measure Transport}

We first introduce a class of maps that are compositions of gradients of
convex functions and push the uniform forward to arbitrary distributions.

\begin{definition}[Composition Measure Transport]
\label{def:CMK} Let $\mu $ be the uniform distribution on~$[0,1]^{d}$, and
let~$P$ be an arbitrary distribution on $\mathbb{R}^{d}$. A composition
measure transport from~$\mu$ to~$P$ is a map~$T: [0,1]^d\rightarrow\mathbb{R}%
^d$ satisfying the following properties.

\begin{enumerate}
\item The map~$T$ pushes~$\mu$ to~$P$, i.e.,~$T\#\mu =P$.

\item There exist l.s.c. convex functions~$\psi_1,\ldots,\psi_{L-1},\psi_L,$
for some~$L$, such that 
\begin{equation*}
T:=\nabla \psi _{L}\circ \nabla \psi _{L-1}\circ ...\circ \nabla \psi
_{1},\;\mu\mbox{-almost everywhere}.
\end{equation*}

\item If~$P$ is absolutely continuous with support in a convex set~$\mathcal{%
V}$ in~$\mathbb{R}^{d}$, then 
\begin{equation*}
T^\ast=\nabla \psi _{1}^{\ast }\circ \nabla \psi _{2}^{\ast }\circ ...\circ
\nabla \psi _{L}^{\ast }
\end{equation*}
exists, equals~$T^{-1}$, $P$-almost everywhere, and satisfies~$%
T^{\ast}\#P=\mu $.
\end{enumerate}
\end{definition}

Existence of composition measure transports (hereafter composition MT) is
guaranteed by Proposition~\ref{prop:polar}. When~$L=1$, the maps $T$ and $%
T^{\ast}$ in Definition~\ref{def:CMK} reduce to vector quantiles and ranks
of Definition~\ref{def:MK}. By allowing~$L$\ to be larger than~$1$, we are
able to choose convex functions $\psi _{l},$ $l\leq L$\ such that the
composition MT maps have explicit expressions. Composition MT maps are the
tools we use to map multivariate marginal distributions into multivariate
uniform distribution to remove all within vector dependence and marginal
information. This is achieved with the following proposition, whose proof
only requires a very minor variation on the proof of Theorem~\ref{thm:Sklar}.

\begin{proposition}
\label{Ethm:Sklar} For any joint distribution~$P$ on~$\mathbb{R}%
^{d_{1}}\times \ldots \times \mathbb{R}^{d_{K}}$ with absolutely continuous
marginals~$P_{k}$ on~$\mathbb{R}^{d_{k}}$ with supports in convex sets, and
any~composition MT maps~$T_{k}$, each~$k\leq K$, there exists a unique
vector copula~$C$ such that the following properties hold.

\begin{enumerate}
\item For any collection~$(A_{1},\ldots ,A_{K})$, where~$A_{k}$ is a Borel
subset of~$\mathbb{R}^{d_{k}}$, $k\leq K$, 
\begin{equation*}
P\left( A_{1}\times \ldots \times A_{K}\right) =P_{C}\left( T_{1}^{\ast
}\left( A_{1}\right) \times \ldots \times T_{K}^{\ast }\left( A_{K}\right)
\right).  \label{ESklar}
\end{equation*}

\item For all Borel sets $B_{1},\ldots ,B_{K}$, in $\mathcal{U}_{1},\ldots ,%
\mathcal{U}_{K}$, 
\begin{equation}
P_{C}\left( B_{1}\times \ldots \times B_{K}\right) =P\left( T_{1}\left(
B_{1}\right) \times \ldots \times T_{K}\left( B_{K}\right) \right) .
\label{CVCopula}
\end{equation}
\end{enumerate}
\end{proposition}

\begin{definition}
\label{def:VC2} The vector copula identified in Proposition~\ref{Ethm:Sklar}
using the composition MT maps~$T_{k}$, $k\leq K$, is called the $\left(
T_{1},...,T_{K}\right)$-vector copula derived from~$P$.
\end{definition}

When $T_{k}$ is the vector quantile of $P_{k}$ for each~$k\leq K$, the $%
\left( T_{1},...,T_{K}\right) $-vector copula derived from $P$ is the vector
copula associated with~$P$ of Definition~\ref{def:VCa}. When at least one of
the composition MT maps~$T_{k}$, $k\leq K$ is not a vector quantile, the $%
\left( T_{1},...,T_{K}\right) $-vector copula derived from~$P$ is not the
vector copula associated with~$P$, but can be used to construct multivariate
distributions~$P^{\prime }$ via part (4) of the vector Sklar Theorem \ref%
{thm:Sklar} such that the vector copula associated with~$P^{\prime }$ is the
prespecified $\left( T_{1},...,T_{K}\right) $-vector copula derived from~$P$%
. Since compositions of gradients of convex functions are in general not
gradients of convex functions themselves, a $\left( T_{1},...,T_{K}\right) $%
-vector copula derived from $P$ does not characterize rank dependence
between random vectors~$\left( Y_{1},...,Y_{K}\right)$ with joint
distribution~$P$, unless all the maps~$T_{1},...,T_{K}$ happen to be
gradients of convex functions.

Proposition \ref{Ethm:Sklar} presents a general approach to constructing
parametric families of vector copulas from parametric families of
multivariate distributions~$P$ with absolutely continuous marginals~$P_{k}$
on~$\mathbb{R}^{d_{k}}$, $k\leq K$. A critical step in this approach is to
derive composition MT map~$T_{k}$ of the marginal distribution~$P_{k}$, $%
k\leq K$ such that~$T_{k}$ has an explicit expression. Below we illustrate
this approach for general elliptical distributions to construct elliptical $%
\left( T_{1},...,T_{K}\right) $-vector copulas.


\subsection{Elliptical $\left( T_{1},...,T_{K}\right) $-Vector Copulas}

\label{sec:EVC}

We first present explicit expressions for composition MT maps of general
elliptical distributions and then present explicit expressions for $\left(
T_{1},...,T_{K}\right) $-vector copulas derived from elliptical
distributions.

\begin{definition}[Elliptical Distributions]
\label{def:ell} A (regular) elliptical distribution on~$\mathbb{R}^{d}$ is
the distribution of a random vector~$R\Sigma^{1/2}U^{(d)}$, where~$R\geq0$
is a radial random variable,~$\Sigma$ is a full rank~$d\times d$ scale
matrix, $U^{(d)}$ is uniform on the unit sphere~$\mathcal{S}^{d-1}$, and~$R$
and~$U^{(d)}$ are mutually independent.
\end{definition}

Examples of elliptical distributions include the following (see e.g.,
Chapter 3 of \citet{NFE:2005}):

\begin{enumerate}
\item \textit{The centered multivariate Gaussian distribution}~$N(0,\Sigma )$%
. It is the distribution of a random vector~$R\Sigma ^{1/2}U^{(d)}$, where~$%
R\sim\chi_{[d]}$,~$\Sigma $ is a full rank~$d\times d$ variance-covariance
matrix, ~$U^{(d)}$ is uniform on the unit sphere~$\mathcal{S}^{d-1}$, and~$R$
and~$U^{(d)}$ are mutually independent;

\item \textit{Student's t distribution}. The multivariate $t_{\nu ,\Sigma }$
with degrees of freedom~$\nu $ and scale~matrix $\Sigma $ is the
distribution of a random vector~$R\Sigma ^{1/2}U^{(d)}$, where~~$R\geq 0$, $%
R^{2}/d$ is a random variable with an~$F_{\nu ,d}$ distribution,~$U^{(d)}$
is uniform on the unit sphere~$\mathcal{S}^{d-1}$, and~$R$ and~$U^{(d)}$ are
mutually independent.
\end{enumerate}

\begin{definition}[Elliptical $\left( T_{1},...,T_{K}\right) $-Vector Copulas%
]
\label{def:EVC} The~$\left( T_{1},...,T_{K}\right) $-vector copula derived
from an elliptical distribution~$P$ on $\mathbb{R}^{d_{1}}\times \ldots
\times \mathbb{R}^{d_{K}}$ is called \emph{elliptical }$\left(
T_{1},...,T_{K}\right) $-\emph{vector copula derived from~$P$}. Proposition %
\ref{Ethm:Sklar} guarantees that the elliptical $\left(
T_{1},...,T_{K}\right) $-derived from a specific elliptical distribution~$P$
exists and is uniquely defined.
\end{definition}

Explicit expressions for elliptical $\left( T_{1},...,T_{K}\right) $-vector
copulas rely on explicit expressions for the composition MT maps $\left(
T_{1},...,T_{K}\right) $. To construct these maps, we rely on Lemma~\ref%
{lemma:radial-W2} in Appendix~\ref{app:OT}.

\begin{lemma}[Elliptical Composition MT]
\label{lemma:MK-ell} Let~$P$ be the elliptical distribution of the vector~$%
\tilde{R}\Sigma ^{1/2}U^{(d)}$, where~$\tilde{R}\geq0$ is a radial random
variable with absolutely continuous distribution,~$\Sigma $ is a full rank~$%
d\times d$ scale matrix, $U^{(d)}$ is uniform on the unit sphere~$\mathcal{S}%
^{d-1}$, independent of~$\tilde{R}$. The map~$T$ defined for every~$u$ in~$%
[0,1]^d$ by 
\begin{eqnarray}
T\left( u\;;\tilde R,\Sigma\right) =\frac{F_{\tilde{R}}^{-1}\circ
F_{R}\left( \left\Vert \Phi ^{-1}\left( u\right) \right\Vert \right) }{%
\left\Vert \Phi ^{-1}\left( u\right) \right\Vert }\Sigma ^{1/2}\Phi
^{-1}\left( u\right),  \label{eq:EVQ}
\end{eqnarray}
where~$R\sim\chi_{[d]}$, and~$\Phi^{-1}$ uses the componentwise notation of~(%
\ref{eq:Phi-1}), is a composition MT (with~$L=3$) that pushes the uniform~$%
\mu$ to~$P$.
\end{lemma}

We now combine Lemma~\ref{lemma:MK-ell} and Proposition~\ref{Ethm:Sklar} to
characterize elliptical $\left( T_{1},...,T_{K}\right) $-vector copulas.

\begin{lemma}[Characterization of Elliptical $\left( T_{1},...,T_{K}\right) $%
-Vector Copulas]
\label{lemma:EVC} \mbox{}\vskip1pt\mbox{} The~$\left( T_{1},...,T_{K}\right) 
$-vector copula derived from the elliptical distribution~$P$ of a vector~$%
\tilde{R}\Sigma ^{1/2}U^{(d)}$ is characterized by~(\ref{CVCopula}) where~$%
T_k(\cdot):=T(\cdot\;;\tilde R,\Sigma_k)$, as in~(\ref{eq:EVQ}), and~$\Sigma
_{k}$ denotes the~$k$-th diagonal block of~$\Sigma $ for all~$k\leq K$.
\end{lemma}

Explicit expressions for the composition MT maps $T_{k}$ of elliptical
distributions $P_{k}$ can be derived from Lemma \ref{lemma:MK-ell} and the
corresponding elliptical copulas can be obtained from Lemma~\ref{lemma:EVC}.
The latter also provides a generic procedure to simulate a random vector
distributed according to a prescribed elliptical $\left(
T_{1},...,T_{K}\right) $-vector copula:

Step 1. Generate a random vector~$Y=(Y_1,\ldots,Y_K)$ from the prescribed
elliptical distribution;

Step 2. Let~$U=(T_1^\ast(Y_1),\ldots,T_K^\ast(Y_K))$, where~$T_k$ is given
in Lemma~\ref{lemma:EVC}. Vector~$U$ is distributed according to the desired
elliptical $\left( T_{1},...,T_{K}\right) $-vector copula.

We present two examples of elliptical $\left( T_{1},...,T_{K}\right) $%
-vector copulas below.


\begin{example}[Gaussian $\left( T_{1},...,T_{K}\right) $-Vector Copulas]

For any integer $d$, a centered Gaussian distribution on~$\mathbb{R}^{d}$ is
the distribution of a random vector~$\widetilde{R}\Sigma ^{1/2}U^{(d)}$,
where $\widetilde{R}\sim \mathcal{\chi }_{\left[ d\right] }$, $U^{(d)}$ is
uniform on the unit sphere~$\mathcal{S}^{d-1}$, and~$\widetilde{R}$ and~$%
U^{(d)}$ are mutually independent. Thus, the Gaussian $\left(
T_{1},...,T_{K}\right) $-vector copula can be constructed using Lemma~\ref%
{lemma:EVC}. In the Gaussian case, the composition MT maps from Lemma~\ref%
{lemma:EVC} are~$T_{k}\left( u\right) =\Sigma _{k}^{1/2}\Phi ^{-1}\left(
u\right) $ and the Gaussian $\left( T_{1},...,T_{K}\right) $-vector copula
is the distribution function of $\left( T_{1}^{\ast }\left( Y_{1}\right)
,...,T_{K}^{\ast }\left( Y_{K}\right) \right) $, where $T_{k}^{\ast }\left(
z\right) :=\Phi \left( z\right) \Sigma _{k}^{1/2}$ for $z\in \mathbb{R}%
^{d_{k}}$ in which~$\Phi $ uses the componentwise notation of~(\ref{eq:Phi-1}%
) and~$\Sigma _{k}$ is the variance-covariance matrix of $Y_{k}$.

Since $\left( \Sigma _{1}^{-1/2}Y_{1},...,\Sigma _{K}^{-1/2}Y_{K}\right)
\sim \Phi _{d}\left( \cdot ;\Omega \right) $, where 
\begin{equation*}
\Omega =\left( 
\begin{tabular}{llll}
$I_{d_{1}}$ & $\Sigma _{1}^{-1/2}\Sigma _{12}\Sigma _{2}^{-1/2}$ & $\cdots $
& $\Sigma _{1}^{-1/2}\Sigma _{1K}\Sigma _{K}^{-1/2}$ \\ 
$\Sigma _{2}^{-1/2}\Sigma _{21}\Sigma _{1}^{-1/2}$ & $I_{d_{2}}$ & $\cdots $
& $\Sigma _{2}^{-1/2}\Sigma _{2K}\Sigma _{K}^{-1/2}$ \\ 
$\vdots $ & $\vdots $ & $\ddots $ & $\vdots $ \\ 
$\Sigma _{K}^{-1/2}\Sigma _{K1}\Sigma _{1}^{-1/2}$ & $\Sigma
_{K}^{-1/2}\Sigma _{K2}\Sigma _{2}^{-1/2}$ & $\cdots $ & $I_{d_{K}}$%
\end{tabular}%
\right) ,  \label{Omega}
\end{equation*}%
we obtain that 
\begin{equation*}
C^{Ga}\left( u_{1},\ldots ,u_{K};\Omega \right) =\Phi _{d}\left( \Phi
^{-1}\left( u_{1}\right) ,\ldots ,\Phi ^{-1}\left( u_{K}\right) ;\Omega
\right) .
\end{equation*}%
This is the Gaussian vector copula presented in Example~\ref{ex:Gaussian}.
For each $k\leq K$, the (classical) copula of~$\Sigma _{k}^{-1/2}Y_{k}$ is
the independence copula and the vector copula~$C^{Ga}$ measures the
between-dependence structure in 
\begin{equation*}
\left( \Sigma _{1}^{-1/2}Y_{1},\ldots ,\Sigma _{K}^{-1/2}Y_{K}\right).
\end{equation*}
However,~$Y_k$ and~$\Sigma _{k}^{-1/2}Y_{k}$ are not comonotonic unless~$%
\Sigma _{k}=\sigma _{k}^{2}I_{d_{k}}$ for a scalar~$\sigma _{k}^{2}>0$, and
the comonotonic invariance in Theorem~\ref{thm:CI} does not apply. Hence,
when~$\Sigma\ne\Omega$, the $\left( T_{1},...,T_{K}\right)$-vector copula
derived from~$N\left( 0,\Sigma \right)$ according to Definition~\ref{def:VC2}
may not be the vector copula associated with~$N\left( 0,\Sigma \right)$
according to Definition~\ref{def:VCa}.
\end{example}


\begin{example}[Student's t $\left( T_{1},...,T_{K}\right) $-vector copulas]


A zero mean Student's~$t$ distribution with~$\nu $ degrees of freedom and
scale matrix~$\Sigma $ on~$\mathbb{R}^{d}$ is characterized by~$Q\Sigma
^{1/2}U^{\left( d\right) }$, where~$Q\geq 0,$ $Q^{2}/d$ follows an~$F$
distribution with~$\left( d,\nu \right) $ degrees of freedom. For $k\leq K$,
let 
\begin{equation}
T_{k}\left( u_{k}\right) =\frac{F_{Q_{k}}^{-1}\circ F_{R_{k}}\left(
\left\Vert \Phi ^{-1}\left( u_{k}\right) \right\Vert \right) }{\left\Vert
\Phi ^{-1}\left( u_{k}\right) \right\Vert }\Sigma _{k}^{1/2}\Phi ^{-1}\left(
u_{k}\right) ,  \label{eq:Student}
\end{equation}%
where $Q_{k}\geq 0$, $Q_{k}^{2}/d_{k}\sim F_{d_{k},\nu }$ and~$R_{k}\sim 
\mathcal{X}_{\left[ d_{k}\right] }$. The $\left( T_{1},...,T_{K}\right) $%
-vector copula density derived from the centered Student's~$t$ distribution
with $\nu $ degrees of freedom and scale matrix $\Sigma $ on~$\mathbb{R}^{d}$%
, where $d=d_{1}+\ldots +d_{K}$, is 
\begin{equation*}
c^{t}\left( u_{1},\ldots ,u_{K};\Sigma ,\nu \right) =t_{d}\left( T_{1}\left(
u_{1}\right) ,\ldots ,T_{K}\left( u_{K}\right) ;\Sigma ,\nu \right)
\prod_{k=1}^{K}\left[ t_{d_{k}}\left( T_{k}\left( u_{k}\right) ;\Sigma
_{k},\nu \right) \right] ^{-1},
\end{equation*}%
where~$t_{d}\left( \cdot ;\Sigma ,\nu \right) $ denotes the density of
Student's $t$ on~$\mathbb{R}^{d}$ with scale matrix $\Sigma $ and degree of
freedom $\nu $. Although in general, Student's t $\left(
T_{1},...,T_{K}\right) $-vector copula is not the vector copula associated
with Student's t distribution, it is identical to the (classical) Student's
t copula, when~$d_{k}=1$ for each~$k\leq K$, as shown in Appendix~\ref%
{app:stud}.

Finally, we provide an algorithm for simulation of Student's $t$ vector
copula. It generalizes Algorithm 5.10 in \citet{NFE:2005} for simulation of
Student's $t$ copulas.

Step 1. Generate $Z\sim N_{d}\left( 0,\Sigma \right) $;

Step 2. Generate a variable $W\sim Ig\left( \frac{\nu }{2},\frac{\nu }{2}%
\right) $ independently and let $Y=\sqrt{W}Z;$

Step 3. The random vector~$U=\left( T_{1}^{-1}\left( Y_{1}\right) ,\ldots
,T_{K}^{-1 }\left( Y_{K}\right) \right) $, where~$T_k$ is given in~(\ref%
{eq:Student}), and~$Y=(Y_1,\ldots,Y_K)$, follows distribution~$C^{t}\left(
\cdot\;;\Sigma ,\nu \right) $.
\end{example}


\subsection{Kendall Vector Copulas}

\label{sec:Kendall}

A second method to construct parametric vector copulas, inspired by
hierarchical Kendall copulas in \citet{Brechmann:2014}, is to make use of
the stochastic representation of a vector copula we establish below. We
illustrate this method by introducing a new class of parametric vector
copulas. It turns out to be the subclass of hierarchical Kendall copulas
with independence cluster copulas. We thus call them Kendall vector copulas.

\begin{proposition}[Stochastic Representation of Vector Copula]
\label{SD-VC} Let $U:=$ $\left( U_{1},...,U_{K}\right) $ be a random vector
of dimension~$d:=d_{1}+...+d_{K}$ with each $U_{k}\sim \mu _{k}$ for $k\leq
K $. Then 
\begin{equation*}
U\overset{d}{=}\left( \exp \left( R_{1}U_{1}^{\left( d_{1}\right) }\right)
,...,\exp \left( R_{K}U_{K}^{\left( d_{K}\right) }\right) \right) ,
\end{equation*}%
where for $k\leq K,$ $U_{k}^{\left( d_{k}\right) }$ is uniform on the unit
simplex on $\mathbb{R}^{d_{k}}$ independent of the random variable $R_{k}$
with distribution function%
\begin{equation}
F_{R_{k}}\left( x\right) =\exp \left( -x\right) \sum_{j=0}^{d_{k}-1}\frac{%
x^{j}}{j!}\text{ for }x\in (-\infty ,0].  \label{RDistribution}
\end{equation}
\end{proposition}

\begin{definition}[Kendall Vector Copulas]
\label{KVC} We call a \emph{Kendall vector copula with nesting copula} $%
C_{0}:\left[ 0,1\right] ^{K}\rightarrow \left[ 0,1\right] $ the distribution
of~$U$ in Proposition~\ref{SD-VC}, in the case where~$U_{1}^{\left(
d_{1}\right) },...,U_{K}^{\left( d_{K}\right) }$ are mutually independent, $%
(U_{1}^{\left( d_{1}\right) },...,U_{K}^{\left( d_{K}\right) })$ is
independent of~$(R_{1},...,R_{K})$, and the (classical) copula of $%
(R_{1},...,R_{K})$ is $C_{0}$.
\end{definition}

Since $U_{1}^{\left( d_{1}\right) },...,U_{K}^{\left( d_{K}\right) }$ are
mutually independent, the dependence structure in a Kendall vector copula
denoted as $C_{KV}$ is characterized by the nesting copula $C_{0}$. We show
in Appendix~\ref{app:Ken-Closed} that the class of Kendall vector copulas
defined in Definition \ref{KVC} is the subclass of hierarchical Kendall
copulas introduced in \citet{Brechmann:2014} with independence cluster
copulas. When $C_{0}$ is absolutely continuous, this allows us to derive the
following closed form expression for the density of Kendall vector copulas: 
\begin{eqnarray}
c_{KV}\left( u_{1},...,u_{K}\right) &=&c_{0}\left( F_{R_{1}}\left( \ln \left[
\Pi _{j=1}^{d_{1}}u_{1j}\right] \right) ,...,F_{R_{K}}\left( \ln \left[ \Pi
_{j=1}^{d_{K}}u_{Kj}\right] \right) \right)  \notag \\
&&\hskip100pt\times \Pi _{k=1}^{K}\left( \Pi _{j=1}^{d_{k}}u_{kj}\right) ,
\label{KVC-density}
\end{eqnarray}%
where $c_{0}$ is the density of $C_{0}$.

A special class of Kendall vector copulas is obtained when the nesting
copula $C_{0}$ is Archimedean (see Appendix~\ref{app:arch}). By choosing
Clayton, Gumbel, and Frank copula generators, we obtain different classes of
Kendall vector copulas with different between vector dependence structure.

Finally we present an algorithm to simulate from a Kendall vector copula
with nesting copula $C_{0}$ based on the stochastic representation of a
Kendall vector copula in Proposition~\ref{SD-VC}. It is similar to
Algorithms 14 and 20 in \citet{Brechmann:2014}.

Step 1. Generate~$\left( V_{1},...,V_{K}\right) $ from~$C_{0}.$

Step 2. Let~$R_{k}=F_{R_{k}}^{-1}\left( V_{k}\right) $ for each~$k\leq K$.

Step 3. Generate mutually independent~$U_{k}^{\left( d_{k}\right)}$ from the
uniform distribution on the unit simplex on~$\mathbb{R}^{d_{k}}$ for~$k\leq
K $.

Step 4. Let~$U_{k}=\left( U_{k1},\ldots ,U_{kd_{k}}\right) $, with $%
U_{kj}=\exp \left( R_{k}U_{kj}^{\left( d_{k}\right) }\right) $ for~$%
j=1,...,d_{k}$ and~$k\leq K$. Then~$U=\left( U_{1},...,U_{K}\right) $
follows the Kendall vector copula~$C_{KV}$ with nesting copula~$C_{0}.$

To illustrate the Kendall vector copula construction and its possible uses,
we fit Kendall copulas with Clayton, Frank, Gaussian or Gumbel nesting
copulas and investigate the effect of the financial crisis on between-vector
dependence in our vector of five international stock indices. As before,~$%
R_{1}(Y_{1})$ and~$R_{2}(Y_{2})$ are the population vector ranks of~$Y_{1}$
(Hang Seng and Nikkei) and~$Y_{2}$ (FTSE, S\&P and DAX) respectively. Assume 
$(Y_{1},Y_{2})$ have Kendall vector copula given in (\ref{KVC-density}). The
nesting copula $C_{0}$, i.e., the copula of~$(R_{1},R_{2})$, is either
Clayton, Frank, Gaussian or Gumbel. Formulas for the latter are given in
Appendix~\ref{app:cop-list} for convenience. Estimation is straightforward,
as the densities of the Kendall copula with Clayton, Frank, Gaussian, and
Gumbel nesting copulas are given in closed form. 
\begin{figure}[tbp]
\centering
\includegraphics[width=1.7in]{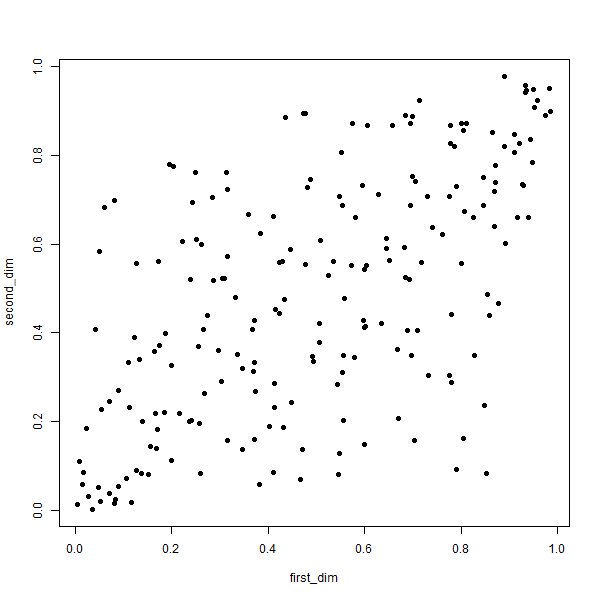} %
\includegraphics[width=1.7in]{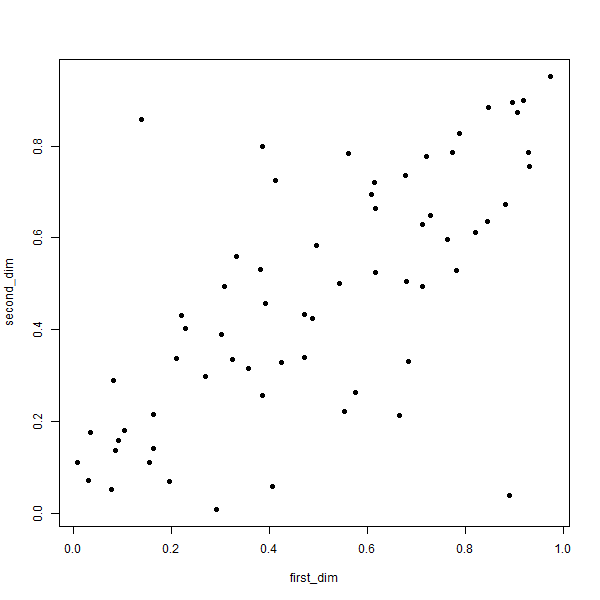} %
\includegraphics[width=1.7in]{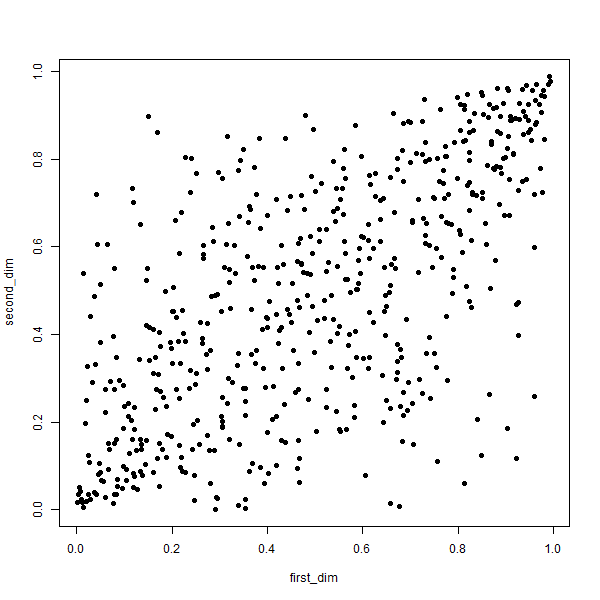}
\caption{Scatterplot of~$(\hat{V}_{1},\hat{V}_{2})$ for pre-crisis, crisis
and post-crisis index returns respectively.}
\label{fig:Kendall}
\end{figure}

Figure~\ref{fig:Kendall} shows bivariate scatterplots for the empirical
version~$(\hat{V}_{1},\hat{V}_{2})$ of~$%
(V_{1},V_{2}):=(F_{R_{1}}(R_{1}),F_{R_{2}}(R_{2}))$ for the pre-crisis,
crisis and post-crisis periods respectively. Appendix B.9 shows 
that~$V_{1}=K_{2}(U_{1}U_{2})$ and~$V_{2}=K_{3}(U_{3}U_{4}U_{5})$%
, where~$(U_{1},U_{2})$ and $(U_{3},U_{4},U_{5})$
denote the components of the ranks~$R_{1}(Y_{1})$ and~$%
R_{2}(Y_{2})$ respectively and~$K_{k}$ denotes the
Kendall distribution of the independence copula, whose expression is given
in~(\ref{eq:K-Dist}). The sample~$(\hat{V}_{1i},\hat{V}_{2i})_{i=1}^{n}$ is
computed from the sample of empirical vector ranks via~$\hat{V}_{1i}:=K_{2}(%
\hat{U}_{1i}\hat{U}_{2i})$ and~$\hat{V}_{2i}:=K_{3}(\hat{U}_{3i}\hat{U}_{4i}%
\hat{U}_{5i})$, where~$(\hat{U}_{1i},\hat{U}_{2i})$ and~$(\hat{U}_{3i},\hat{U%
}_{4i},\hat{U}_{5i})$ denote the components of the empirical ranks~$\hat{R}%
_{1}(Y_{1i})$ and~$\hat{R}_{2}(Y_{2i})$ respectively. Clayton, Gaussian,
Frank and Gumbel bivariate copulas are fitted to the sample~$(\hat{V}_{1i},%
\hat{V}_{2i})_{i=1}^{n}$ and estimated parameters are transformed to the
corresponding value of Kendall's~$\tau $ for comparison. Results are given
in Table~\ref{table:Kendall}. The vector copula analysis of these five stock
indices thereby yields an insight into the financial contagion that
accompanied the 2008 financial crisis. Kendall's~$\tau $ increased during
the crisis, then declined again, but remained weakly higher in the
post-crisis than in the pre-crisis periods.

\begin{table}[tbph]
\caption{Kendall's~$\protect\tau $ implied from the Kendall nesting copula
parameters estimated from the stock index residuals.}
\label{table:Kendall}\centering
\begin{tabular}{lccc}
\hline\hline
&  &  &  \\ 
&  &  &  \\ 
& Pre-Crisis & Crisis & Post-Crisis \\ 
&  &  &  \\ \hline
&  &  &  \\ 
&  &  &  \\ 
Clayton & 0.36 & 0.38 & 0.36 \\ 
&  &  &  \\ 
Frank & 0.45 & 0.56 & 0.49 \\ 
&  &  &  \\ 
Gaussian & 0.44 & 0.46 & 0.46 \\ 
&  &  &  \\ 
Gumbel & 0.47 & 0.57 & 0.50 \\ 
&  &  &  \\ \hline\hline
\end{tabular}%
\end{table}



\section*{Concluding remarks}

We have proposed a framework to characterize dependence between random
vectors, as distinct from within vector dependence in the same way copulas
characterize dependence between random variables as distinct from marginal
information. The basic building block is what we call a vector copula, which
is simply a multivariate distribution with uniform multivariate marginals.
Hence, the class of vector copulas is merely a subclass of the class of
copulas, with the added constraint that the multivariate marginals are
uniform. The contribution of the paper was to associate one such vector
copula with any multivariate distribution with given non overlapping
marginals and to show that this vector copula does indeed characterize
between vector dependence for such a distribution. By characterization of
between vector dependence, we mean that the original distribution can be
recovered from the vector copula and the multivariate marginals, the vector
copula is unique in case of absolutely continuous multivariate marginals,
and finally, that the vector copula is invariant to the class of
transformations that leave between vector dependence unchanged.

The main device we used to derive the vector copula associated with a
distribution is the multivariate analogue of the probability integral
transform, or vector ranks, introduced in \citet{CGHH:2017} and based on the
theory of measure transportation. However, since vector ranks are rarely
available in explicit form, parametric families of vector copulas cannot be
readily constructed from existing parametric families of multivariate
distributions. We therefore proposed a couple of strategies to construct
parametric families of vector copulas. One is based on ideas inherited from
the literature on hierarchical copula models, in which the lower level
copulas are independence copulas. Another relies on measure transport theory
to turn parametric families of multivariate distributions into vector
copulas by transforming the multivariate marginals into uniforms. To obtain
explicit forms, we rely on compositions of optimal transport maps. However,
since the composition of optimal transport maps is not, in general, an
optimal transport map itself, the parametric family of vector copulas thus
obtained is not, in general, the family of vector copulas associated with
the multivariate distribution used to derive it. In other words, it does not
characterize between vector dependence in the latter. It is nonetheless a
useful tool to construct new families of multivariate distributions in a way
that properly distinguishes between vector dependence and within vector
dependence, a goal that cannot be achieved with existing parametric families
of multivariate distributions.

Throughout the paper, we have illustrated the use of semiparametric
vector copula based models to study financial contagion through the
evolution of between-vector dependence before, during and after the 2008
financial crisis for a collection of five aggregate stock indices.
Developing formal estimation and inference methods for such models is of
utmost importance and is an on-going project of the authors. Once developed,
we anticipate many applications of this new tool, that mirror applications
of traditional copulas in quantitative finance and econometrics.



\appendix

\section{Vector ranks and multivariate monotonicity}

\label{app:OT}


\subsection{Convex analysis preliminaries}

Let~$\mathcal{U}$ and~$\mathcal{Y}$ be convex subsets of~$\mathbb{R}^{d}$
with non-empty interiors. A \emph{convex} function~$\psi $ on~$\mathcal{U}$
refers to a function~$\psi :\mathcal{U}\rightarrow \mathbb{R}\cup \{+\infty
\}$ for which~$\psi ((1-t)x+tx^{\prime })\leq (1-t)\psi (x)+t\psi (x^{\prime
})$ for any~$(x,x^{\prime })$ such that~$\psi (x)$ and~$\psi (x^{\prime })$
are finite and for any~$t\in (0,1)$. Such a function is continuous on the
interior of the convex set dom~$\psi :=\{x\in \mathcal{U}:\psi (x)<\infty \}$%
, called the domain of~$\psi $. A function~$\psi $ is called \emph{proper}
if~dom$(\psi )\neq \varnothing $.

A convex function~$\psi$ is differentiable Lebesgue-almost everywhere in dom~%
$\psi $, by the Alexandroff Theorem, see Theorem~14.25 in %
\citet{Villani:2009}. For any proper function~$\psi :\mathcal{U}\mapsto 
\mathbb{R}\cup \{+\infty \}$, the \emph{conjugate}~$\psi ^{\ast }:\mathcal{Y}%
\mapsto \mathbb{R}\cup \{+\infty \}$ of~$\psi $ is defined for each~$y\in 
\mathcal{Y}$ by 
\begin{equation*}
\psi ^{\ast }(y):=\sup_{z\in \mathcal{U}}[y^{\top }z-\psi (z)]. 
\end{equation*}
The conjugate~$\psi ^{\ast }$ of~$\psi $ is a proper convex
lower-semi-continuous function on $\mathcal{Y}$. We shall call a \emph{%
conjugate pair of potentials} over~$(\mathcal{U},\mathcal{Y})$ any pair of
lower-semi-continuous convex functions ~$(\psi ,\psi ^{\ast })$ that are
conjugates of each other. {Legendre-Fenchel duality} states that for any
proper lower semicontinuous convex function~$\psi$, $(\psi^\ast)^\ast=\psi$.

The \emph{subdifferential}~$\partial\psi(x)$ of a convex function~$\psi$ at
a point~$x\in$~dom$(\psi)$ is 
\begin{equation*}
\partial\psi(x):=\{ h\in\mathbb{R}^d\; : \; \forall y\in\mathbb{R}^d,
\psi(y) \geq \psi(x)+(y-x)^\top h\}. 
\end{equation*}
When~$\psi$ is differentiable at~$x$, then~$\partial\psi(x)=\{\nabla\psi(x)\}
$. The subdifferential is the set 
\begin{equation*}
\partial\psi:=\{(x,y)\in\mathbb{R}^d\times\mathbb{R}^d,
y\in\partial\psi(x)\}. 
\end{equation*}
For any proper lower semicontinuous convex function~$\psi :\mathcal{U}%
\rightarrow \mathbb{R}\cup \{+\infty \}$, and any pair~$x,y\in$~dom$(\psi)$,
we have, by Proposition~2.4 in \citet{Villani:2003}, 
\begin{eqnarray*}
\psi(x)+\psi^\ast(y)=x^\top y & \Leftrightarrow & y\in\partial\psi(x) \\
& \Leftrightarrow & x\in\partial\psi^\ast(y).
\end{eqnarray*}
Therefore, we have (Lemma 3.1 in \citet{GS:2019} : 
\begin{eqnarray}  \label{eq:subdiff}
\begin{array}{c}
y\in\partial\psi(\partial\psi^\ast(y))\mbox{ for all }y\in\mathbb{R}^d, \;
x\in\partial\psi^\ast(\partial\psi(x))\mbox{ for all }x\in\mathbb{R}^d, \\ 
\\ 
\mbox{and }(\partial\psi)^{-1}(B)=\partial\psi^\ast(B), \mbox{ for any Borel
set }B.%
\end{array}%
\end{eqnarray}

A subset~$\Gamma\in\mathbb{R}^d\times\mathbb{R}^d$ is called \emph{%
cyclically monotone} if for all~$m\geq2$ and for all sequences of pairs~$%
(x_1,y_1),\ldots,(x_m,y_m)$ in~$\Gamma$, 
\begin{eqnarray*}
\sum_{i=1}^m\|y_i-x_i\|^2\geq\sum_{i=1}^m\|y_{i-1}-x_i\|^2,
\end{eqnarray*}
with the convention~$y_0=y_m$, or equivalently, 
\begin{eqnarray*}
\sum_{i=1}^my_i^\top(x_{i+1}-x_i)\leq0,
\end{eqnarray*}
with the convention~$x_{m+1}=x_1$. A function~$\psi$ is called cyclically
monotone if its graph~$\{(x,y)\in\mathbb{R}^d\times\mathbb{R}^d, y=\psi(x)\}$
is a cyclically monotone set, or equivalently if for all finite collection~$%
x_1,\ldots,x_m$ of points in~dom$(\psi)$, 
\begin{eqnarray*}
\sum_{i=1}^m\psi(x_i)^\top(x_{i+1}-x_i)\leq0,
\end{eqnarray*}
with the convention~$x_{m+1}=x_1$. By Theorem~2.27 of \citet{Villani:2003},~$%
\Gamma$ is cyclically monotone if and only if it is included in the
subdifferential of a proper lower semicontinuous convex function, and the
maximal cyclically monotone subsets of~$\mathbb{R}^d\times\mathbb{R}^d$ are
exactly the subdifferentials of proper lower semicontinuous convex functions
on~$\mathbb{R}^d$. Hence, the cyclically monotone maps on~$\mathbb{R}^d$ are
the gradients of convex functions.


\subsection{Vector ranks}

Let~$X$ be a random variable with absolutely continuous distribution~$P_{X}$%
. The cumulative distribution~$F_{X}$ of~$X$ is a rank function that
transports the distribution~$P_{X}$ to the uniform distribution~$\mu $ on~$%
[0,1]$, in the sense that~$F_{X}(X)\sim \mu $. Conversely, the quantile
function~$F_{X}^{-1}$ transports the distribution~$\mu $ to~$P_{X}$, in the
sense that for any~$U\sim \mu $, $F_{X}^{-1}(U)\sim P_{X}$. Among all maps
that transport~$\mu $ to~$P_{X}$ (resp.~$P_{X}$ to~$\mu $), the map~$%
F_{X}^{-1}$ (resp.~$F_{X}$) is very special: it is the only monotone
nondecreasing one. The map~$F_{X}^{-1}$ is also the almost everywhere only
map that minimizes expected squared error (or, equivalently, maximizes
correlation) 
\begin{equation}
\label{eq:OT}
\min_{T}\mathbb{E}[\Vert U-T(U)\Vert ^{2}],\mbox{ subject to }U\sim \mu %
\mbox{ and }T(U)\sim P_{X}.
\end{equation}%
Problem~(\ref{eq:OT}) can still be used to define quantiles in~$\mathbb{R}%
^{d}$ in spite of the absence of a natural ordering of outcomes. Let~$\mu $
now be supported on a compact subset~$\mathcal{K}$ of~$\mathbb{R}^{d}$. The
map~$Q_{X}:\mathcal{K}\rightarrow \mathbb{R}^{d}$ that minimizes~(\ref{eq:OT}%
) exists and is almost everywhere unique and is cyclically monotone (almost
everywhere equals the gradient of a convex function) by the \emph{Polar
Factorization Theorem}, see \citet{RR:90} and \citet{Brenier:91}.

Suppose the hypotheses of Proposition~\ref{prop:polar}(1) hold and~$P_X$ and~%
$\mu$ have finite second moments, then the function $\psi$, or optimal
potential, solves the optimization problem 
\begin{eqnarray}  \label{eq:Kdual}
\int\psi d\mu+\int\psi^\ast dP_X=\inf_{(\tilde\psi, \tilde\psi^*)} \left(
\int\tilde\psi d\mu+\int\tilde\psi^\ast dP_X \right),
\end{eqnarray}
where the infimum is taken over the class of conjugate pairs of potentials.
This problem is dual to the optimal transport problem (\ref{eq:OT}).
Moreover, under the hypotheses of Proposition~\ref{prop:polar}(1), $%
\nabla\psi$ is the unique optimal transport map from~$\mu$ to~$P$ for
quadratic cost, in the sense that any other optimal transport coincides with 
$\nabla \psi$ on a set of $\mu$-measure~$1$ (see \citet{Villani:2003}).
Under the hypotheses of Proposition~\ref{prop:polar}(2), $\nabla\psi^*$ is
the unique optimal (reverse) transport map from $P_X$ to~$\mu$ for quadratic
cost, in the sense that any other optimal transport coincides with $\nabla
\psi^*$ on a set of $P_X$-measure~$1$.

Formulation~(\ref{eq:Kdual}) shows that transport potentials~$%
(\psi,\psi^\ast)$ can be computed as the solution of a convex optimization
problem. The discrete version of~(\ref{eq:Kdual}), when both~$\mu$ and~$P_X$
are discretized, is a linear programming problem. Hence, empirical versions
of the transport potentials can be computed efficiently. See \citet{PC:2019}
for a recent account of computational optimal transport. A recent algorithm
not included in the latter is \citet{JL:2020}. Convergence of empirical
versions to population counterparts is shown in~\citet{CGHH:2017}, %
\citet{ZP:2019}, \citet{BCHM:2020}, \citet{GS:2019}, \citet{BSS:2018}, and %
\citet{DS:2019}.

There are very few known explicit formulas for optimal transport maps. Some
notable exceptions we use in this paper are given in the following lemma. If~%
$X$ is a random vector with distribution~$P_X$ and~$T$ solves~(\ref{eq:OT}),
then the pair~$(X,T(X))$ is called an \emph{optimal coupling with respect to
quadratic transportation cost}.

\begin{lemma}[\citet{CRT:93}]
\label{lemma:radial-W2} Let~$A$ be a positive definite matrix on~$\mathbb{R}%
^d$. Let~$\alpha $ be a non-decreasing function and define $T(x)=\alpha
(\lVert x\rVert)x/\lVert x\rVert$. If~$X$ has absolutely continuous
distribution on~$\mathbb{R}^d$, then $(X,AX)$ and~$(X,T(X))$ are both
optimal couplings with respect to quadratic transportation cost.
\end{lemma}

Problem~(\ref{eq:OT}) requires existence of second moments, a restriction
that is incompatible with a definition of vector quantiles for arbitrary
distributions. However, in a seminal extension of the Polar Factorization
Theorem based on cyclical monotonicity of the solution, \citet{McCann:95}
shows existence and uniqueness of the cyclically monotone map transporting
an absolutely continuous distribution~$\mu$ to an arbitrary distribution~$P_X
$, as stated in Proposition~\ref{prop:polar}. This motivates the definition
of vector quantiles and ranks in Definition~\ref{def:MK}.

When both~$\mu$ and~$P_X$ admit densities~$m$ and~$p$ respectively, the
change of variables formula is a special case of the celebrated \emph{%
Monge-Amp\`ere equation}, see Chapter~4 of \citet{Villani:2003}. From
Proposition~\ref{prop:polar}, we know that there exists an almost everywhere
unique gradient~$\nabla\psi$ (the vector quantile of Definition~\ref{def:MK}%
) of a convex function~$\psi$ such that~$\nabla\psi\#\mu=P_X$. Hence, for
any continuous bounded function~$f$, we have 
\begin{eqnarray}  \label{eq:MA1}
\int f(x)p(x)dx & = & \int f(\nabla\psi(u))m(u)du.
\end{eqnarray}
Assuming that~$p$ is positive and that~$\psi$ is twice continuously
differentiable and strictly convex, then the change of variables~$%
x=\nabla\psi(u)$ in the left-hand side of~(\ref{eq:MA1}) yields 
\begin{eqnarray}  \label{eq:MA2}
\int f(x)p(x)dx & = & \int
f(\nabla\psi(u))p(\nabla\psi(u))\det(D^2\psi(u))du.
\end{eqnarray}
From equality of the right-hand sides of~(\ref{eq:MA1}) and~(\ref{eq:MA2})
and the fact that~$f$ can be chosen arbitrarily, we get the Monge-Amp\`ere
equation 
\begin{eqnarray}  \label{eq:MA3}
\det\left( D^2\psi(u) \right) & = & \frac{m(u)}{p\left( \nabla\psi(u)\right)}%
, \;\;\mu\mbox{-a.e.},
\end{eqnarray}
which we use to compute vector copula densities in Section~\ref{sec:Sklar}.

Finally, the following lemma shows how vector quantiles and ranks relate to
the quantiles and ranks of mutually independent subvectors.

\begin{lemma}[\citet{GS:2019}]
\label{lemma:GS} Let~$X_k$ be a random vector on~$\mathbb{R}^{d_k}$ with
vector quantile~$Q_k$ and vector rank~$R_k$,~$k\leq K$, and let~$%
X=(X_1,\ldots,X_K)$ have vector quantile~$Q$ and vector rank~$R$. Then the
following statements are equivalent:

\begin{enumerate}
\item The random vectors~$X_1,\ldots,X_K$ are mutually independent.

\item We have~$Q=(Q_1,\ldots,Q_K)$ and~$R=(R_1,\ldots,R_K)$ almost
everywhere.
\end{enumerate}
\end{lemma}


\section{Proofs of results in the main text}

\label{app:proofs}


\subsection{A Proof of the Vector Sklar Theorem}

For the proof of Theorem~\ref{thm:Sklar}, we need the following definition
and result, due to \citet{Vorobev:62} and \citet{Kellerer:64}.

\begin{definition}[Decomposability]
A finite collection~$\{S_{1},\ldots,S_{N}\}$ of subsets of a finite set~$%
\mathcal{S}$ is called \emph{decomposable} if there exists a permutation~$%
\sigma$ of~$\{1,\ldots,N\}$ such that 
\begin{eqnarray}
\left(\,\bigcup_{l<m}S_{\sigma(l)}\right)\cap
S_{\sigma(m)}\in\bigcup_{l<m}2^{S_{\sigma(l)}}\;\mbox{ for all }1\leq m\leq
N.  \label{eq:dec}
\end{eqnarray}
\end{definition}

For instance, the collection of subsets~$\{\{1,2\},\{2,3\}\}$ of set~$%
\{1,2,3\}$ is decomposable, but~$\{\{1,2\},\{2,3\},\{1,3\}\}$ is not.%
\footnote{%
We are grateful to an anonymous referee for suggesting adding this
counterexample.}

Proposition~\ref{prop:Kel} below is proved in \citet{Kellerer:64}. It is
also referenced in \citet{Rusch:2010} as Theorem~1.27 page~29 and in
Section~3.7 of \citet{Joe:97}.

\begin{proposition}[Existence of probability measures with overlapping
marginals]
\label{prop:Kel}

Let~$X_{k}:=\left(\mathbb{R},\mathcal{B}\left(\mathbb{R}\right)\right)$, for~%
$k=1,\ldots,K$. Let~$\mathcal{S}:=\{S_{1},\ldots,S_{N}\} $ be an arbitrary
collection of subsets of~$\{1,\ldots,K\}$. For each~$j\leq N$, let~$P_{j}$
be a probability measure on the product space~$\times_{k\in S_{j}}X_{i}$.
Then there exists a probability measure on~$\times_{k}X_{k}$ with marginal~$%
P_{j}$ on~$\times_{k\in S_{j}}X_{k}$, all~$j=1,\ldots,N$, if the following
two conditions hold.

\begin{enumerate}
\item The marginals~$P_{j_{1}}$ and $P_{j_{2}}$ coincide on~$\times_{k\in
S_{j_{1}}\cap S_{j_{2}}}X_{k}$, all~$j_{1}<j_{2}\leq N$.

\item The collection~$\mathcal{S}$ is decomposable.
\end{enumerate}
\end{proposition}

\begin{proof}[Proof of Theorem~\protect\ref{thm:Sklar}]
\begin{enumerate}
\item For~$i=1,\ldots ,2\times\sum_{k=1}^{K}d_{k},$ define~$X_{i}:=(\mathbb{R%
},\mathcal{B}(\mathbb{R}))$. Let~$S_{1}:=\left\{ 1,2,\ldots
,\sum_{k=1}^{K}d_{k}\right\} $ and for each~$k\leq K$, set~$d_{0}:=0$, and
define 
\begin{equation*}
S_{1+k}:=\left\{ \sum_{l=\min \{1,k-1\}}^{k-1}d_{l}+1,\ldots
,\sum_{l=1}^{k}d_{l},\sum_{l=1}^{K}d_{l}+\sum_{l=1}^{k-1}d_{l}+1\ldots
,\sum_{l=1}^{K}d_{l}+\sum_{l=1}^{k}d_{l}\right\} .
\end{equation*}%
We first show existence of a joint probability distribution~$\pi $ on~$%
\times _{i}X_{i}$ with marginals~$P$ on~$\times _{i\in S_{1}}X_{i}$ and~$(%
\mbox{Id},T_{k})\#\mu _{k}$ on~$\times _{i\in S_{1+k}}X_{i}$, each~$k\leq K$%
. For this, we only need to verify conditions~(1) and~(2) of Proposition~\ref%
{prop:Kel} applied to the family~$\mathcal{S}:=\{S_{1},\ldots ,S_{N}\}$,
where~$N:=1+K$. Condition~(1) is satisfied, since the marginal of~$(\mbox{Id}%
,T_{k})\#\mu _{k}$ is~$T_{k}\#\mu _{k}$, which is equal to~$P_{k}$ by
definition of the generalized vector quantile (Definition~\ref{def:MK}).
There remains to show that the collection~$\mathcal{S}$ is decomposable.
Take any integer~$m\in \{2,\ldots ,N\}$. We have 
\begin{equation*}
\left( \,\bigcup_{l<m}S_{l}\right) \cap S_{m}=\left\{ \sum_{l=\min
\{1,k-1\}}^{k-1}d_{l}+1,\ldots ,\sum_{l=1}^{k}d_{l}\right\} ,
\end{equation*}%
which belongs to the set of subsets of~$S_{1}$, so that Equation~(\ref%
{eq:dec}) holds. The theorem follows since the projection of~$\pi $ on~$%
\mathcal{U}_{1}\times \ldots \times \mathcal{U}_{K}$ is a vector copula as
desired.

\item Let $\left( Y_{1},\ldots ,Y_{K},U_{1},\ldots ,U_{K}\right) $ follow a
joint distribution on~$(\mathbb{R}^{d_{1}}\times \ldots \times \mathbb{R}%
^{d_{K}})\times (\mathcal{U}_{1}\times \ldots \times \mathcal{U}_{K})$ as
in~(1). Then for each~$k\leq K$, $\left( U_{k},Y_{k}\right) \sim (\mbox{Id}%
,\nabla\psi_{k})\#\mu _{k}$ and~$Y_{k}=\nabla\psi_{k}\left( U_{k}\right) $.
Hence for all Borel sets~$A_{1},\ldots ,A_{K}$, in~$\mathbb{R}%
^{d_{1}},\ldots ,\mathbb{R}^{d_{K}}$, 
\begin{eqnarray*}
P\left( A_{1}\times \ldots \times A_{K}\right) &=&\Pr \left(
\nabla\psi_{1}\left( U_{1}\right) \in A_{1},\ldots ,\nabla\psi_{K}\left(
U_{K}\right) \in A_{K}\right) \\
&=&\Pr \left( U_{1}\in \partial \psi _{1}^{\ast } ,\ldots , U_{K}\in
\partial \psi _{K}^{\ast }\left( A_{K}\right) \right) \\
&=&P_{C}\left( \partial\psi_{1}^{\ast }\left( A_{1}\right) \times \ldots
\times \partial\psi_{K}^{\ast }\left( A_{K}\right) \right) ,
\end{eqnarray*}%
where the penultimate equality follows from Equation~(\ref{eq:subdiff}).

\item When $P_{k}$ is absolutely continuous with support in a convex set, $%
U_{k}=\nabla\psi_{k}^{\ast}\left( Y_{k}\right) $, $P_{k}$-almost everywhere,
by Definition~\ref{def:MK}. So for all Borel sets $B_{1},\ldots ,B_{K}$, in $%
\mathcal{U}_{1},\ldots ,\mathcal{U}_{K}$, 
\begin{eqnarray*}
P_{C}\left( B_{1}\times \ldots \times B_{K}\right) &=&\Pr \left(
\nabla\psi_{1}^{\ast}\left( Y_{1}\right) \in B_{1},\ldots
,\nabla\psi_{K}^{\ast}\left( Y_{K}\right) \in B_{K}\right) \\
&=&\Pr \left( Y_{1}\in \nabla\psi_{1}\left( B_{1}\right) ,\ldots ,Y_{K}\in
\nabla\psi_{K}\left( B_{K}\right) \right) \\
&=&P\left( \nabla\psi_{1}\left( B_{1}\right) \times \ldots \times
\nabla\psi_{K}\left( B_{K}\right) \right) ,
\end{eqnarray*}%
implying that the vector copula~$C$ is uniquely determined.
\end{enumerate}
\end{proof}


\subsection{Proof of Theorem~\protect\ref{thm:CI}}

Let $\left( Y_{1},\ldots ,Y_{K},U_{1},\ldots ,U_{K}\right) $ follow a joint
distribution on $(\mathbb{R}^{d_{1}}\times \ldots \times \mathbb{R}%
^{d_{K}})\times (\mathcal{U}_{1}\times \ldots \times \mathcal{U}_{K})$ as in
Theorem~\ref{thm:Sklar}(1). For each~$k\leq K$, denote by~$Q_{k} $ and~$%
\tilde{Q}_{k}$ vector quantiles associated with the distributions of~$Y_{k}$
and~$\tilde{Y}_{k}$, respectively. Then, for each~$k\leq K$, it holds that $%
Y_{k}=Q_{k}(U_{k})$, $\mu _{k}$-almost surely. Since~$Y_{k}$ and~$\tilde{Y}%
_{k}$ are comonotonic, we also have $\tilde{Y}_{k}=\tilde{Q}_{k}(U_{k})$, $%
\mu _{k}$-almost surely. Hence the joint distribution of $(\tilde{Y}%
_{1},\ldots ,\tilde{Y}_{K},U_{1},\ldots ,U_{K})$ satisfies the conditions
that characterize a vector copula associated with the distribution~$\tilde{P}
$ of $(\tilde{Y}_{1},\ldots ,\tilde{Y}_{K})$.


\subsection{Proof of Lemma~\protect\ref{lemma:antitone}}

Let the vector $\left( Y_{1},\ldots ,Y_{K},U_{1},\ldots ,U_{K}\right) $
follow a joint distribution on $(\mathbb{R}^{d_{1}}\times \ldots \times 
\mathbb{R}^{d_{K}})\times (\mathcal{U}_{1}\times \ldots \times \mathcal{U}%
_{K})$ as in Theorem~\ref{thm:Sklar}(1). For each~$k\leq K$, denote by~$%
Q_{k} $ and~$\tilde{Q}_{k}$ vector quantiles associated with the
distributions of~$Y_{k}$ and~$\tilde{Y}_{k}$, respectively. Then $%
Y_{k}=Q_{k}(U_{k})$, $\mu _{k}$-almost surely. Since~$Y_{k}$ and~$\tilde{Y}%
_{k}$ are comonotonic (resp. countermonotonic) for each~$k\leq K_{1}$ (resp. 
$K_1<k\leq K$), we have $\tilde{Y}_{k}=\tilde{Q}_{k}(U_{k})$ (resp. $\tilde{Y%
}_{k}=\tilde{Q}_{k}(1_{d_{k}}-U_{k})$) for each~$k\leq K_{1}$ (resp. $%
K_1<k\leq K$), $\mu _{k}$-almost surely. Hence the joint distribution of $%
(U_{1},\ldots ,U_{K_{1}},1_{d_{K_{1}+1}}-U_{K_{1}+1},\ldots ,1_{d_{K}}-U_{K})
$ satisfies the conditions that characterize a vector copula associated with
the distribution~$\tilde{P}$ of $(\tilde{Y}_{1},\ldots ,\tilde{Y}_{K})$.


\subsection{Proof of Lemma~\protect\ref{lemma:CC}}

Let~$U_{1},\ldots ,U_{K}$ be comonotonic vectors. Then, by definition of
comonotonicity, in view of the fact that~$U_{1}$ has distribution~$\mu $, $%
U_{k}=Q_{k}(U_{1})$, where~$Q_{k}$ is the vector quantile of~$U_{k}$ for
each~$2\leq k\leq K$. Since~$U_{k}$ also has distribution~$\mu $, it follows
that~$Q_{k}=\mbox{Id}$, for each~$2\leq k\leq K$. Hence,~$U_{1}=\cdots =U_{K}
$. The conclusion on comonotonic copulas follows. Let~$U_{1}$ and~$U_{2}$ be
distributed according to~$\mu $ and be countermonotonic. Then~$%
U_{2}=Q_{2}(1_{d}-U_{1})$, where~$Q_{2}$ is the vector quantile of~$U_{2}$.
Since~$U_{2}$ also has distribution~$\mu $, it follows that~$T_{2}=\mbox{Id}$%
. Hence,~$U_{2}=1_{d}-U_{1}$. The characterization of the countermonotonic
copulas follows.


\subsection{Proof of Lemma~\protect\ref{lemma:MK-ell}}

First,~$T=\nabla \psi _{3}\circ \nabla \psi _{2}\circ \nabla \psi _{1}$,
where%
\begin{eqnarray*}
\nabla \psi _{1}\left( u\right) = \Phi ^{-1}\left( u\right), \; \nabla\psi
_{2}\left( u\right) = \frac{F_{\tilde{R}}^{-1}\circ F_{R}(\lVert u\rVert )}{%
\lVert u\rVert}u, \mbox{ and } \nabla \psi _{3}\left( u\right) = \Sigma
^{1/2}u.
\end{eqnarray*}

It follows from Lemma~\ref{lemma:GS} in Appendix~\ref{app:OT} that $\nabla
\psi _{1}$ is the gradient of a convex function pushing the uniform $\mu $
to $N\left( 0,I_{d}\right) $. Lemma~\ref{lemma:radial-W2} in Appendix~\ref%
{app:OT} implies that both $\nabla \psi _{2}$ and~$\nabla \psi _{3}$ are
gradients of convex functions. We now show that~$\nabla \psi _{3}\circ
\nabla \psi _{2}$ pushes~$N\left( 0,I_{d}\right) $ forward to $P$. Recall
that~$N\left( 0,I_{d}\right) $ is the distribution of~the random vector $%
RU^{(d)}$, where~$R\sim \chi _{\left[ d\right] }$, and~$U^{(d)}$ is uniform
on the unit sphere~$\mathcal{S}^{d-1}$, independent of~$R$. Thus for~$Z\sim
N\left( 0,I_{d}\right) $, it holds that 
\begin{eqnarray*}
\nabla \psi _{3}\circ \nabla \psi _{2}\left( Z\right)  &=&\frac{F_{\tilde{R}%
}^{-1}\circ F_{R}(\lVert Z\rVert )}{\lVert Z\rVert }\Sigma ^{1/2}Z \\
&\overset{d}{=}&\Sigma ^{1/2}F_{\tilde{R}}^{-1}\circ F_{R}(\lVert
~RU^{(d)}\rVert )\frac{RU^{(d)}}{\lVert ~RU^{(d)}\rVert }.
\end{eqnarray*}%
Corollary 3.23 in \citet{NFE:2005} implies that 
\begin{equation*}
\left( \lVert ~RU^{(d)}\rVert ,\frac{RU^{(d)}}{\lVert ~RU^{(d)}\rVert }%
\right) \overset{d}{=}\left( R,U^{(d)}\right) ,
\end{equation*}%
which in turn implies that%
\begin{eqnarray*}
\left( F_{\tilde{R}}^{-1}\circ F_{R}(\lVert ~RU^{(d)}\rVert ),\frac{RU^{(d)}%
}{\lVert ~RU^{(d)}\rVert }\right)  &\overset{d}{=}&\left( F_{\tilde{R}%
}^{-1}\circ F_{R}\left( R\right) ,U^{(d)}\right)  \\
&\overset{d}{=}&\left( \tilde{R},U^{(d)}\right) .
\end{eqnarray*}%
Finally we obtain that $\nabla \psi _{3}\circ \nabla \psi _{2}\left(
Z\right) \overset{d}{=}$~$\tilde{R}\Sigma ^{1/2}U^{(d)}$.

\subsection{Proof of Lemma~\protect\ref{lemma:EVC}}

Noting that $P_{k}$ is the distribution of~$\tilde{R}\Sigma
_{k}^{1/2}U^{(d_{k})}$, $k\leq K$, Lemma~\ref{lemma:EVC} follows from
Theorem~\ref{thm:Sklar}(2) and Lemma~\ref{lemma:MK-ell}.

\subsection{Student's t vector copula reduces to classical Student's t copula%
}

\label{app:stud}

When $d_{k}=1$ for all~$k\leq K$, $Q_{k}^{2}/d_{k}=Q_{k}^{2}\sim F_{1,\nu }$
and $R_{k}\sim \chi _{\left[ 1\right] }$. Then for $Z\sim N\left( 0,1\right) 
$, we obtain that 
\begin{align*}
T_{k}\left( u_{k}\right) & =\nabla \psi _{k}\left( \Phi ^{-1}\left(
u_{k}\right) \right) =\frac{F_{Q_{k}}^{-1}\circ F_{R_{k}}\left( \left\vert
\Phi ^{-1}\left( u_{k}\right) \right\vert \right) }{\left\vert \Phi
^{-1}\left( u_{k}\right) \right\vert }\Phi ^{-1}\left( u_{k}\right) \Sigma
_{k}^{1/2} \\
& =F_{Q_{k}}^{-1}\circ \Pr \left( \left\vert Z\right\vert \leq \left\vert
\Phi ^{-1}\left( u_{k}\right) \right\vert \right) \frac{\Phi ^{-1}\left(
u_{k}\right) }{\left\vert \Phi ^{-1}\left( u_{k}\right) \right\vert }\Sigma
_{k}^{1/2} \\
& =F_{Q_{k}}^{-1}\circ \Pr \left( -\left\vert \Phi ^{-1}\left( u_{k}\right)
\right\vert \leq Z\leq \left\vert \Phi ^{-1}\left( u_{k}\right) \right\vert
\right) \frac{\Phi ^{-1}\left( u_{k}\right) }{\left\vert \Phi ^{-1}\left(
u_{k}\right) \right\vert }\Sigma _{k}^{1/2}.
\end{align*}%
Let $S_{k}$ follow the Student's $t$ distribution with $\nu $ degrees of
freedom. By the relation between the Student's $t$ distribution and the $F$
distribution, it holds that $S_{k}^{2}\sim F_{1,\nu }$. Since $Q_{k}^{2}\sim
F_{1,\nu }$, $Q_{k}$ can be characterized as $\left\vert S_{k}\right\vert $.
Thus, we have that for $q_{k}\geq 0$,%
\begin{align*}
F_{Q_{k}}\left( q_{k}\right) & =\mathrm{Pr}\left( 0\leq Q_{k}\leq
q_{k}\right) =\mathrm{Pr}\left( 0\leq \left\vert S_{k}\right\vert \leq
q_{k}\right)  \\
& =\mathrm{Pr}\left( -q_{k}\leq S_{k}\leq q_{k}\right) =2\left(
F_{S_{k}}\left( q_{k}\right) -1/2\right) ,
\end{align*}%
where $F_{S_{k}}\left( \cdot \right) $ is the distribution function of $S_{k}
$. For $u_{k}\geq 1/2$, it holds that%
\begin{align*}
T_{k}\left( u_{k}\right) & =F_{Q_{k}}^{-1}\circ \Pr \left( -\Phi ^{-1}\left(
u_{k}\right) \leq Z\leq \Phi ^{-1}\left( u_{k}\right) \right) \Sigma
_{k}^{1/2} \\
& =F_{Q_{k}}^{-1}\left( 2\left( u_{k}-1/2\right) \right)
=F_{S_{k}}^{-1}\left( u_{k}\right) \Sigma _{k}^{1/2},
\end{align*}%
and for $u_{k}<1/2$, it holds that 
\begin{align*}
T_{k}\left( u_{k}\right) & =-F_{Q_{k}}^{-1}\circ \Pr \left( \Phi ^{-1}\left(
u_{k}\right) \leq Z\leq -\Phi ^{-1}\left( u_{k}\right) \right) \Sigma
_{k}^{1/2} \\
& =-F_{Q_{k}}^{-1}\left( 2\left( 1/2-u_{k}\right) \right)
=F_{S_{k}}^{-1}\left( u_{k}\right) \Sigma _{k}^{1/2}.
\end{align*}%
Thus, letting $Q_{\nu }$ and $t\left( \cdot ;\nu \right) $ denote the
quantile and density functions of the Student's $t$ distribution with $\nu $
degrees of freedom respectively, we obtain the vector copula density given by%
\begin{equation*}
c^{t}\left( u_{1},\ldots ,u_{K};\Sigma ,\nu \right) =t_{d}\left( Q_{\nu
}\left( u_{1}\right) ,\ldots ,Q_{\nu }\left( u_{K}\right) ;\Omega ,\nu
\right) \prod_{k=1}^{K}\left[ \frac{1}{t\left( Q_{\nu }\left( u_{k}\right)
;\nu \right) }\right] ,
\end{equation*}%
which is the same as the density of the classical $t$ copula, where 
\begin{equation*}
\Omega =\left( 
\begin{tabular}{llll}
$1$ & $\Sigma _{12}/\left( \Sigma _{1}^{1/2}\Sigma _{2}^{1/2}\right) $ & $%
\cdots $ & $\Sigma _{1K}/\left( \Sigma _{1}^{1/2}\Sigma _{K}^{1/2}\right) $
\\ 
$\Sigma _{21}/\left( \Sigma _{1}^{1/2}\Sigma _{2}^{1/2}\right) $ & $1$ & $%
\cdots $ & $\Sigma _{2K}/\left( \Sigma _{K}^{1/2}\Sigma _{2}^{1/2}\right) $
\\ 
$\vdots $ & $\vdots $ & $\ddots $ & $\vdots $ \\ 
$\Sigma _{K1}/\left( \Sigma _{1}^{1/2}\Sigma _{K}^{1/2}\right) $ & $\Sigma
_{K2}/\left( \Sigma _{K}^{1/2}\Sigma _{2}^{1/2}\right) $ & $\cdots $ & $1$%
\end{tabular}%
\right) .
\end{equation*}


\subsection{Proof of Proposition~\protect\ref{SD-VC}}

\label{app:arch}

A $K$ dimensional copula is called Archimedean if it permits the
representation:%
\begin{equation*}
C\left( u\right) =\psi \left( \psi ^{-1}\left( u_{1}\right) +...+\psi
^{-1}\left( u_{K}\right) \right) ,u\in \left[ 0,1\right] ^{K}
\end{equation*}%
for some Archimedean generator $\psi $, namely, $\psi :[0,\infty
)\rightarrow \lbrack 0,1]$ is a nonincreasing and continuous function which
satisfies the conditions: $\psi (0)=1$ and $\lim_{x\rightarrow \infty }\psi
(x)=0$ and is strictly decreasing on $[0,\inf \{x:\psi (x)=0\})$.

It is known that the distribution of $U_{k}$ is the $d_{k}$-dimensional
independence copula and is Archimedean with generator function~$\exp \left(
-u\right)$. Theorem~3.1(ii) and Example~3.2 in \citet{NN:2009} imply that 
\begin{equation}
\left( -\ln \left( U_{k1}\right) ,...,-\ln \left( U_{kd_{k}}\right) \right) 
\overset{}{=}-R_{k}U_{k}^{\left( d_{k}\right) },  \label{SRforU}
\end{equation}%
for some uniform~$U_k^{(d_k)}$ on the unit simplex, and some~$R_{k}$ with
distribution given in (\ref{RDistribution}), independent of~$U_k^{(d_k)}$,
each~$k\leq K$. Inverting the expression in (\ref{SRforU}) yields the result.


\subsection{Derivation of the Kendall vector copula density~(\protect\ref%
{KVC-density})}

\label{app:Ken-Closed}

To see why~(\ref{KVC-density}) holds, let $\widetilde{V}_{k}:=K_{k}\left(
\Pi _{j=1}^{d_{k}}U_{kj}\right) $, where $K_{k}$ is the Kendall distribution
function of the independence copula, i.e.:%
\begin{equation}
K_{k}\left( u\right) :=\mathbb{P}\left( \prod_{j=1}^{d_{k}}U_{kj}\leq
u\right) =u\sum_{i=0}^{d_{k}-1}\frac{\log ^{i}\left( 1/u\right) }{i!}.
\label{eq:K-Dist}
\end{equation}%
The nesting copula of a hierarchical Kendall copula in \citet{Brechmann:2014}
is the copula of $\left( \widetilde{V}_{1},...,\widetilde{V}_{K}\right) $.
Now, since $\widetilde{V}_{k}=K_{k}\left( \exp \left( R_{k}\right) \right) $%
, see Remark~5 in \citet{Brechmann:2014}, we have $F_{R_{k}}(\ln u)=\mathbb{P%
}(\widetilde{V}_{k}\leq K_{k}(u))=K_{k}(u)$, by definition of the Kendall
distribution. It implies that 
\begin{equation*}
V_{k}=F_{R_{k}}\left( R_{k}\right) =F_{R_{k}}\left( \ln \left( \exp \left(
R_{k}\right) \right) \right) =K_{k}\left( \exp \left( R_{k}\right) \right) =%
\widetilde{V}_{k}.
\end{equation*}%
Applying Theorem 8 in \citet{Brechmann:2014} to Kendall vector copulas, we
obtain~(\ref{KVC-density}) as desired. 


\subsection{Some classical parametric copulas}

\label{app:cop-list}

For convenience, we list formulas for the four bivariate copulas we use in
Section~\ref{sec:dm}. See for instance \citet{Nelsen:2006}.

\begin{itemize}
\item Bivariate Gaussian copula with parameter~$\rho$: 
\begin{equation*}
C^{\mbox{\scriptsize Ga}}(u_1,u_2;\rho)=\Phi_2(\Phi^{-1}(u_1),%
\Phi^{-1}(u_2);\rho).
\end{equation*}
Relation to Kendall's~$\tau$: $\tau=\frac{2}{\pi}\arcsin\rho$.

\item Bivariate Clayton copula with parameter~$\theta\in[-1,\infty)%
\backslash\{0\}$: 
\begin{equation*}
C^{\mbox{\scriptsize Cl}}(u_1,u_2,\theta):=\left[ \max\{ u_1^{-\theta} +
u_2^{-\theta}-1,0 \} \right]^{-\frac{1}{\theta}}.
\end{equation*}
Relation to Kendall's~$\tau$: $\tau=\frac{\theta}{\theta+2}$.

\item Bivariate Frank copula with parameter~$\theta\in\mathbb{R}%
\backslash\{0\}$: 
\begin{equation*}
C^{\mbox{\scriptsize Fr}}(u_1,u_2,\theta):=-\frac{1}{\theta} \log\left[ 1+ 
\frac{\left( e^{-\theta u}-1\right)\left( e^{-\theta v}-1\right)}{%
e^{-\theta}-1}\right].
\end{equation*}
Relation to Kendall's~$\tau$: $\tau=1-\frac{4}{\theta}\left(1-\frac{1}{\theta%
}\int_0^\theta\frac{t}{e^t-1}dt\right)$.

\item Bivariate Gumbel copula with parameter~$\theta\in[1,\infty)$: 
\begin{equation*}
C^{\mbox{\scriptsize Gu}}(u_1,u_2,\theta):=\exp\left[ -\left( (-\ln
u_1)^{\theta} + (-\ln u_2)^{\theta} \right)^{\frac{1}{\theta}} \right].
\end{equation*}
Relation to Kendall's~$\tau$: $\tau=\frac{\theta-1}{\theta}$.
\end{itemize}



\bibliographystyle{abbrvnat}
\bibliography{copula}


\end{document}